\documentclass[reqno]{amsart}
\usepackage[a4paper]{geometry}
\usepackage[T1]{fontenc}
\usepackage[utf8]{inputenc}
\usepackage{amssymb}
\usepackage[mathscr]{eucal}
\usepackage{mathtools}
\usepackage{enumerate}
\usepackage{accents}
\usepackage{color}
\usepackage[colorlinks=true]{hyperref}
\hypersetup{urlcolor=blue,citecolor=red,linkcolor=blue}
\usepackage[initials]{amsrefs}
\numberwithin{equation}{section}
\theoremstyle{plain}
\newtheorem{theorem}{Theorem}[section]
\newtheorem{proposition}[theorem]{Proposition}
\newtheorem{corollary}[theorem]{Corollary}
\theoremstyle{definition}
\newtheorem*{rh-pb*}{Basic RH problem}
\newtheorem*{sol-rh-pb*}{Soliton RH problem}
\newtheorem*{data*}{Data of this RH problem associated with $\BS{u_0(x)}$}
\theoremstyle{remark}
\newtheorem{remark}[theorem]{Remark}
\newtheorem*{notations*}{Notations}
\providecommand{\BS}[1]{\boldsymbol{#1}}
\providecommand{\D}[1]{\mathbb{#1}}
\newcommand{\dd}{\mathrm{d}}
\newcommand{\eul}{\mathrm{e}}
\newcommand{\ii}{\mathrm{i}}
\newlength{\dhatheight}
\newcommand{\doublehat}[1]{%
    \settoheight{\dhatheight}{\ensuremath{\hat{#1}}}%
    \addtolength{\dhatheight}{-0.3ex}%
    \hat{\vphantom{\rule{1pt}{\dhatheight}}%
    \smash{\hat{#1}}}}
\providecommand{\abs}[1]{\lvert#1\rvert}
\providecommand{\accol}[1]{\lbrace#1\rbrace}

\renewcommand{\Im}{\operatorname{Im}}
\newcommand{\ord}{\mathrm{O}}
\DeclareMathOperator{\Res}{Res}
\begin{document}
\title[Riemann--Hilbert approach for the mCH equation]{A Riemann--Hilbert approach to the modified Camassa--Holm equation with nonzero boundary conditions}
\author[A.~Boutet de Monvel]{Anne Boutet de Monvel}
\address{Institut de Math\'ematiques de Jussieu-Paris Rive Gauche,
Universit\'e de Paris,
8 place Aur\'elie Nemours, case 7012,
75205 Paris Cedex 13,
France}
\email{anne.boutet-de-monvel@imj-prg.fr}
\author[I.~Karpenko]{Iryna Karpenko}
\address{B.I.~Verkin Institute for Low Temperature Physics and Engineering,
47 Nauky Avenue, 61103 Kharkiv, Ukraine}
\email{inic.karpenko@gmail.com}
\author[D.~Shepelsky]{Dmitry Shepelsky}
\address{B.I.~Verkin Institute for Low Temperature Physics and Engineering,
47 Nauky Avenue, 61103 Kharkiv, Ukraine\\
V.N.~Karazin Kharkiv National University, 4 Svobody Square, 61022 Kharkiv, Ukraine}
\email{shepelsky@yahoo.com}
\date{\today}
\begin{abstract}
The paper aims at developing the Riemann--Hilbert problem approach to the modified Camassa--Holm (mCH) equation in the case when the solution is assumed to approach a non-zero constant at the both infinities of the space variable. In this case, the spectral problem for the associated Lax pair equation has a continuous spectrum, which allows formulating the inverse spectral problem as a Riemann--Hilbert factorization problem with jump conditions across the real axis. We obtain a representation for the solution of the Cauchy problem for the mCH equation and also a description of certain soliton-type solutions, both regular and non-regular.
\end{abstract}
\maketitle
\tableofcontents
\section{Introduction}\label{sec:1}

The Camassa--Holm (CH) equation \cites{CH93,CHH94}
\begin{equation}\label{CH}
u_t - u_{xxt}+ 3u u_x - 2u_x u_{xx} - u u_{xxx} = 0,
\end{equation}
which can also be written in terms of the momentum variable
\begin{equation}\label{CH-1}
m_t+\left(u m\right)_x + u_x m = 0,\quad m\coloneqq u-u_{xx},
\end{equation}
has been studied intensively over the last 25 years, due to its rich mathematical structure. It is a model for the unidirectional propagation of shallow water waves over a flat bottom \cites{J02,CL09}, is bi-Hamiltonian \cite{CH93}, and is completely integrable with algebro-geometric solutions \cite{Q03}. The local and global well-posedness of the Cauchy problem for the CH equation have been studied extensively \cites{CE98-1, CE98-2, D01}. In particular, it has both globally strong solutions  and blow-up solutions at finite time \cites{C00,CE98-1,CE98-2,CE98-3}, and also it has globally weak solutions in $H^1(\D{R})$ \cites{BC07, CM00, XZ00}.

The soliton-type solutions of the CH equation vanishing at infinity \cite{CHH94} are weak solutions, having the form of peaked waves ($u(x,t)$ and $u_x(x,t)$ are bounded but $u_x(x,t)$ is discontinuous), which are orbitally stable \cite{CS00}.
 
On the other hand, adding to \eqref{CH} a linear dispersion term $b u_x$ with $b>0$ leads to a form of the CH equation 
\begin{equation}\label{CH-b}
u_t-u_{xxt}+b u_x + 3u u_x - 2u_x u_{xx} - u u_{xxx} = 0,
\end{equation}
which supports conventional smooth solitons \cites{C01,BS06,BS08}. 

Over the last few years various modifications and generalizations of the CH equation have been introduced, see, e.g., \cite{YQZ18} and references therein. Novikov \cite{N09} applied the perturbative symmetry approach in order to classify integrable equations of the form 
\[
\left(1-\partial_x^2\right) u_t = F(u, u_x, u_{xx}, u_{xxx}, \dots),\qquad u=u(x,t), \quad \partial_x=\partial/\partial x,
\]
assuming that $F$ is a homogeneous differential polynomial over $\D{C}$, quadratic or cubic in $u$ and its $x$-derivatives (see also \cite{MN02}). In the list of equations presented in \cite{N09}, equation (32), which was the second equation with \emph{cubic} nonlinearity, had the form
\begin{equation}\label{mCH-1}
m_t+\left((u^2-u_x^2)m\right)_x = 0, \quad m\coloneqq u-u_{xx}.
\end{equation}
In an equivalent form, this equation was given by Fokas in \cite{F95} (see also \cite{OR96} and \cite{Fu96}) and has attracted considerable interest since it was re-derived by Qiao \cite{Q06}. So it is sometimes referred to as the Fokas--Olver--Rosenau--Qiao equation \cite{HFQ17}, but is also known as the modified Camassa--Holm (mCH) equation. Equation \eqref{mCH-1} has a bi-Hamiltonian structure \cites{OR96,GLOQ13} and possesses a Lax pair \cite{Q06}. Its algebro-geometric quasiperiodic solutions are studied in \cite{HFQ17}. The local well-posedness for classical solutions and global weak solutions to \eqref{mCH-1} in Lagrangian coordinates are discussed in \cite{GL18}. It also has solitary wave solutions \cite{GLOQ13}
\[
u(x,t)=\frac{p}{2}\eul^{-|x-x(t)|},\quad m(x,t)=p\delta(x-x(t))\quad\text{with }x(t)=\frac{1}{6}p^2t.
\]

Notice that considering the initial value problem for the Camassa--Holm equation with a linear dispersion term \eqref{CH-b} and with initial data decaying to $0$ as $x\to\pm\infty$ is equivalent to considering the CH equation in the form \eqref{CH} on a nonzero background, i.e., with initial data approaching a nonzero constant as $x\to\pm\infty$. A similar situation takes place, for example, for the Degasperis--Procesi equation
\begin{equation}\label{DP-1}
m_t+\left(um\right)_x+2u_xm=0,  \quad m=u-u_{xx},
\end{equation}
which is also an integrable, CH-type equation with quadratic nonlinearity. On the other hand, for other CH-type equations, in particular, for those with cubic nonlinearity, the situation is different: while considering the equation on a nonzero background again leads to problems supporting smooth solitons, changing variables (leading to zero background) results in an equation having different form, which is not equivalent to adding just a linear dispersion term; see, e.g., the case of the Novikov equation \cite{BS16}.

In the present paper, we consider the initial value problem for the mCH equation \eqref{mCH-1}:
\begin{subequations}\label{mCH1-ic}
\begin{alignat}{4}           \label{mCH-1-ic}
&m_t+\left((u^2-u_x^2)m\right)_x=0,&\quad&m\coloneqq u-u_{xx},&\quad&t>0,&\;&-\infty<x<+\infty,\\
&u(x,0)=u_0(x),&&&&&&-\infty<x<+\infty\label{IC},
\end{alignat}
\end{subequations}
assuming that $u_0(x)\to 1$ as $x\to\pm\infty$, and we search for a solution that preserves this behavior: $u(x,t)\to 1$ as $x\to\pm\infty$ for all $t>0$. Then, in analogy with the CH equation and other CH-type equations, one can expect that the Cauchy problem \eqref{mCH1-ic} supports smooth soliton solutions.

Introducing a new function $\tilde u$ by 
\begin{equation}\label{utilde}
u(x,t)=\tilde u(x-t,t)+1,
\end{equation}
the mCH equation reduces to 
\begin{subequations}\label{mCH2}
\begin{align}\label{mCH-2}
&\tilde m_t+\left(\tilde\omega\tilde m\right)_x = 0,\\
\label{tm}
&\tilde m\coloneqq\tilde u-\tilde u_{xx}+1,\\
\label{tom}
&\tilde\omega\coloneqq\tilde u^2-\tilde u_x^2+2\tilde u.
\end{align}
\end{subequations}
In what follows we will study equation \eqref{mCH2} on zero background: $\tilde u\to 0$ as $x\to\pm\infty$. More precisely, we develop the Riemann--Hilbert (RH) problem approach to equation \eqref{mCH-2} on zero background, aiming at obtaining a representation of the solution of the Cauchy problem for \eqref{mCH2} in terms of the solution of an associated RH problem formulated in the complex plane of a spectral parameter. 

The paper is organized as follows. In Section \ref{sec:2} we introduce the Jost solutions of the Lax pair equations written in a form  appropriate for controlling their analytical properties as function of the spectral parameter. In Section \ref{sec:3} we formulate the Riemann--Hilbert problem in two settings: (i) in the original setting, it (implicitly) depends on the physical variables $(x,t)$ as parameters and (ii) in a transformed setting, introducing new variables $(y,t)$ in terms of which the RH problem has an explicit parameter dependence. The data for the later RH problem are uniquely determined by the initial data for the mCH equation, which gives rise to a procedure for solving the Cauchy problem \eqref{mCH1-ic}. In Section \ref{sec:7} we show that starting with the solution of a RH problem with appropriate dependence on the parameters, we always arrive at a solution to the mCH equation, even if the data for this RH problem are not associated with some particular initial data for the mCH equation. Finally, in Section \ref{sec:8}, using the RH problem formalism, we construct smooth as well as non-smooth soliton solutions to the mCH equation. Throughout the text, we emphasize the differences in the implementation of the RH approach to the CH and mCH equations.

\begin{notations*}
Furthermore, $\sigma_1\coloneqq\left(\begin{smallmatrix}0&1\\1&0\end{smallmatrix}\right)$, $\sigma_2\coloneqq\left(\begin{smallmatrix}0&-\ii\\\ii&0\end{smallmatrix}\right)$, and $\sigma_3\coloneqq\left(\begin{smallmatrix}1&0\\0&-1\end{smallmatrix}\right)$ denote the standard Pauli matrices. We also let $f^*(k)\coloneqq\overline{f(\bar k)}$ denote the Schwarz conjugate of a function $f(k)$, $k\in\D{C}$.
\end{notations*}

\section{Lax pairs and eigenfunctions}\label{sec:2}
\subsection{Lax pairs}
In order to deduce the Lax pair for equation \eqref{mCH-2}, we take as starting point the Lax pair for the mCH equation \eqref{mCH-1} \cite{Q06}
\[
\Phi_x=\mathsf{U}\Phi,\qquad\Phi_t=\mathsf{V}\Phi
\]
where $\Phi\equiv\Phi(x,t,\lambda)$, $\mathsf{U}\equiv\mathsf{U}(x,t,\lambda)$, and $\mathsf{V}\equiv\mathsf{V}(x,t,\lambda)$, the coefficients $\mathsf{U}$ and $\mathsf{V}$ being defined by
\begin{align*}
\mathsf{U}&=\frac{1}{2}\begin{pmatrix} -1 & \lambda  m \\
-\lambda m & 1\end{pmatrix},\\
\mathsf{V}&=\begin{pmatrix}\lambda^{-2}+\frac{u^2-u_x^2}{2} &
-\lambda^{-1}(u-u_x)-\frac{\lambda(u^2-u_x^2)m}{2}\\
\lambda^{-1}(u+u_x)+\frac{\lambda(u^2- u_x^2)m}{2} &
		 -\lambda^{-2}-\frac{u^2-u_x^2}{2}\end{pmatrix},
\end{align*}
with $m\coloneqq u-u_{xx}$. This leads us to the pair of equations
\begin{subequations}\label{Lax}
\begin{align}\label{Lax-x}
\Phi_x&=U\Phi,\\
\label{Lax-t}
\Phi_t&=V\Phi,
\end{align}
\end{subequations}
where the coefficients $U\equiv U(x,t,\lambda)$ and $V\equiv V(x,t,\lambda)$ are now defined by
\begin{subequations}\label{Lax-UV}
\begin{align}\label{Lax-U}
U&=\frac{1}{2}\begin{pmatrix} -1 & \lambda \tilde m \\
-\lambda \tilde m & 1 \end{pmatrix},\\
\label{Lax-V}
V&=\begin{pmatrix}\lambda^{-2}+\frac{\tilde\omega}{2} &
-\lambda^{-1}(\tilde u-\tilde u_x+1)-\frac{\lambda\tilde\omega\tilde m}{2}\\
\lambda^{-1}(\tilde u+\tilde u_x+1)+\frac{\lambda\tilde\omega\tilde m}{2} & -\lambda^{-2}-\frac{\tilde\omega}{2}\end{pmatrix}.
\end{align}
\end{subequations}
Here, $\tilde m\coloneqq\tilde u-\tilde u_{xx}+1$ and $\tilde\omega\coloneqq\tilde u^2-\tilde u_x^2+2\tilde u$ as in \eqref{tm} and \eqref{tom}, with $\tilde u$ as in \eqref{utilde}. It can be directly verified that \eqref{mCH-2} is the compatibility condition for the system \eqref{Lax}-\eqref{Lax-UV}. Thus, this system \eqref{Lax}-\eqref{Lax-UV} constitutes a Lax pair for \eqref{mCH-2}.

The RH formalism for integrable nonlinear equations is based on using appropriately defined eigenfunctions, i.e., solutions of the Lax pair, whose behavior as functions of the spectral parameter is well-controlled in the extended complex plane. Notice that the coefficient matrices $U$ and $V$ are traceless, which provides that the determinant of a matrix solution to \eqref{Lax} (composed from two vector solutions) is independent of $x$ and $t$.

Also notice that $U$ and $V$  have singularities (in the extended complex $\lambda$-plane) at $\lambda=0$ and $\lambda=\infty$. In order to control the behavior of solutions to \eqref{Lax} as functions of the spectral parameter $\lambda$ (which is crucial for the Riemann--Hilbert method), we follow a strategy similar to that adopted for the CH equation \cites{BS06,BS08}.

Namely, in order to control the large $\lambda$ behavior of solutions of \eqref{Lax}, we will transform this Lax pair into an appropriate form (see~\cites{BC,BS06,BS08}).

\begin{proposition}
Equation \eqref{mCH-2} admits a Lax pair of the form
\begin{subequations}\label{Lax-Q-form}
\begin{align}
\hat\Phi_x+Q_x\hat\Phi &= \hat U\hat\Phi,\\
\hat\Phi_t+Q_t\hat\Phi &= \hat V\hat\Phi,
\end{align}
\end{subequations}
whose coefficients $Q\equiv Q(x,t,\lambda)$, $\hat U\equiv\hat U(x,t,\lambda)$, and $\hat V\equiv\hat V(x,t,\lambda)$ are $2\times 2$ matrices having the following properties:
\begin{enumerate}[\rm(i)]
\item
$Q$ is diagonal and is unbounded as $\lambda\to\infty$.
\item
$\hat U=\ord(1)$ and $\hat V=\ord(1)$ as $\lambda\to\infty$.
\item
The diagonal parts of $\hat U$ and $\hat V$ decay as $\lambda\to\infty$.
\item
$\hat U\to 0$ and $\hat V\to 0$ as $x\to\pm\infty$.
\end{enumerate}
\end{proposition}

\begin{proof}
We first note that $U$ in \eqref{Lax-U} can be written as
\begin{equation}
\label{U-tilde}
U(x,t,\lambda)=\frac{\tilde m(x,t)}{2}\begin{pmatrix}-1&\lambda\\-\lambda&1\end{pmatrix}+\frac{\tilde m(x,t)-1}{2}\begin{pmatrix}1&0\\
0&-1\end{pmatrix},
\end{equation}
where $\tilde m(x,t)-1\to 0$ as $x\to\pm\infty$. The  first (non-decaying, as $x\to\pm\infty$) term in \eqref{U-tilde} can be diagonalized by introducing
\[
\hat\Phi(x,t,\lambda)\coloneqq D(\lambda)\Phi(x,t,\lambda),
\]
where
\[
D(\lambda)\coloneqq\begin{pmatrix}
1 & - \frac{\lambda}{1+\sqrt{1-\lambda^2}}  \\
- \frac{\lambda}{1+\sqrt{1-\lambda^2}} & 1 \\
\end{pmatrix},
\]
where the square root is chosen so that $\sqrt{1-\lambda^2}\sim\ii\lambda$ as $\lambda\to\infty$. This transforms \eqref{Lax-x} into
\begin{subequations}\label{Lax-1}
\begin{equation}\label{Lax-1-x}
\hat\Phi_x+\frac{\tilde m\sqrt{1-\lambda^2}}{2}\sigma_3\hat\Phi=\hat U \hat\Phi,
\end{equation}
where $\hat U\equiv\hat U(x,t,\lambda)$ is given by
\begin{equation}\label{U-hat}
\hat U=\frac{\lambda(\tilde m-1)}{2\sqrt{1-\lambda^2}}
\begin{pmatrix}
0 & 1 \\
-1 & 0 \\
\end{pmatrix}
+\frac{\tilde m-1}{2\sqrt{1-\lambda^2}}\sigma_3.
\end{equation}
Similarly, the $t$-equation \eqref{Lax-t} of the Lax pair is transformed into
\begin{equation}\label{Lax-1-t}
\hat\Phi_t +\sqrt{1-\lambda^2}
\left(-\frac{1}{2}\tilde m\tilde\omega-\frac{1}{\lambda^2}\right)\sigma_3\hat\Phi= \hat V \hat\Phi,
\end{equation}
where $\hat V\equiv\hat V(x,t,\lambda)$ is given by
\begin{equation}\label{hat-V}
\begin{aligned}
\hat V&=
\frac{1}{2 \sqrt{1-\lambda^2}}\left(\lambda\tilde\omega(\tilde m - 1)
+\frac{2 \tilde u}{\lambda}\right)
\begin{pmatrix}
0 & -1 \\
1 & 0
\end{pmatrix}
+\frac{ \tilde u_x}{\lambda}  \begin{pmatrix}
0 & 1 \\
1 & 0
\end{pmatrix}\\
&\quad-\frac{1}{\sqrt{1-\lambda^2}}\left(\tilde u +\frac{1}{2}(\tilde m-1)\tilde\omega\right) \sigma_3.
\end{aligned}
\end{equation}
\end{subequations}
Now notice that equations \eqref{Lax-1-x} and \eqref{Lax-1-t} have the desired form \eqref{Lax-Q-form}, if we define $Q$ by 
\begin{subequations}\label{Qp}
\begin{equation}\label{Q}
Q(x,t,\lambda)\coloneqq p(x,t,\lambda)\sigma_3, 
\end{equation}
with
\begin{equation}\label{p}
p(x,t,\lambda)\coloneqq -\frac{1}{2}\sqrt{1-\lambda^2}\int_x^{+\infty} (\tilde m(\xi,t)-1)\dd\xi+\frac{\sqrt{1-\lambda^2}}{2}x-\frac{\sqrt{1-\lambda^2}}{\lambda^2}t.
\end{equation}
\end{subequations}
Indeed, $p$ has derivatives
\begin{align*}
p_x&=\frac{\tilde m\sqrt{1-\lambda^2}}{2},\\
p_t&=\sqrt{1-\lambda^2}\left(-\frac{1}{2}\tilde m\tilde\omega-\frac{1}{\lambda^2}\right).
\end{align*}
The first formula is clear, while the second follows from \eqref{mCH-2}.
\end{proof}

\subsection{Eigenfunctions}

The Lax pair in the form \eqref{Lax-1} allows us to determine dedicated solutions having a well-controlled behavior as functions of the spectral parameter $\lambda$ for large values of $\lambda$ via associated integral equations. Indeed, introducing
\begin{equation}\label{zam}
\widetilde\Phi=\hat\Phi\eul^{Q}
\end{equation}
(understanding $\widetilde\Phi$ as a $2\times 2$ matrix), equations \eqref{Lax-1-x} and \eqref{Lax-1-t} can be rewritten as
\begin{equation}\label{comsys}
\begin{cases}
\widetilde\Phi_x+[Q_x,\widetilde\Phi]=\hat U\widetilde\Phi,&\\
\widetilde\Phi_t+[Q_t,\widetilde\Phi]=\hat V\widetilde\Phi,&
\end{cases}
\end{equation}
where $[\,\cdot\,,\,\cdot\,]$ stands for the commutator. We now determine particular (Jost) solutions $\widetilde\Phi_{\pm}\equiv\widetilde\Phi_{\pm}(x,t,\lambda)$ of \eqref{comsys} as solutions of the associated Volterra integral equations:
\begin{equation}\label{inteq}
\widetilde\Phi_{\pm}(x,t,\lambda)=I+\int_{\pm\infty}^x
	\eul^{Q(\xi,t,\lambda)-Q(x,t,\lambda)}\hat U(\xi,t,\lambda)\widetilde\Phi_{\pm}(\xi,t,\lambda)
		\eul^{Q(x,t,\lambda)-Q(\xi,t,\lambda)}\dd\xi,
\end{equation}
that is, taking into account the definition \eqref{Qp} of $Q$,
\begin{equation}\label{eq}
\begin{split}
\widetilde\Phi_+(x,t,\lambda)&=I-\int_x^{+\infty}
\eul^{\frac{\sqrt{1-\lambda^2}}{2}
\int_x^\xi\tilde m(\eta,t)\dd\eta\,\sigma_3}\hat U(\xi,t,\lambda)\widetilde\Phi_+(\xi,t,\lambda)\eul^{-\frac{\sqrt{1-\lambda^2}}{2}\int_x^\xi\tilde m(\eta,t)\dd\eta\,\sigma_3}\dd\xi,\\
\widetilde\Phi_{-}(x,t,\lambda)&=I+\int_{-\infty}^x\eul^{\frac{\sqrt{1-\lambda^2}}{2}\int_x^\xi\tilde m(\eta,t)\dd\eta\,\sigma_3}\hat U(\xi,t,\lambda)\widetilde\Phi_{-}(\xi,t,\lambda)\eul^{-\frac{\sqrt{1-\lambda^2}}{2}\int_x^\xi\tilde m(\eta,t)\dd\eta\,\sigma_3}\dd\xi
\end{split}
\end{equation}
($I$ is the identity matrix). Hereafter, let $\hat\Phi_\pm\coloneqq\widetilde\Phi_\pm\eul^{-Q}$ denote the corresponding Jost solutions of \eqref{Lax-1}.

Introducing a new spectral parameter $k$ by 
\[
\lambda^2=4k^2+1,
\]
the exponentials in \eqref{eq} become $\eul^{\pm\ii k\int_x^\xi \tilde m(\xi,t)\dd\xi\,\sigma_3}$. Moreover, introducing the new space variable 
\begin{equation}\label{shkala}
y(x,t)\coloneqq x-\int_x^{+\infty}(\tilde m(\xi,t)-1)\dd\xi,
\end{equation}
$Q$ takes (by a slight abuse of notations) the form $Q(y,t,k)=-\ii k\left(y-\frac{2t}{4k^2+1}\right)\sigma_3$, which coincides with that in the case of the Camassa--Holm equation \cites{BS06,BS08}.

\begin{remark}
Recall that the pair of renowned integrable equations --- the  Korteweg--de Vries (KdV) equation and the modified Korteweg--de Vries (mKdV) equation --- shares the same $Q$, which, in those cases, has the form $Q(x,t,k)=(\ii kx+4\ii k^3t)\sigma_3$. Therefore, the above consideration gives an additional reason to naming equation \eqref{mCH-1} as the \emph{modified} Camassa--Holm (mCH) equation.
\end{remark}

However, an important difference between the Lax pairs for the CH equation and the mCH equation is that in the latter case, the dependence of the associated coefficient matrix $\hat U(x,t,k)$ (by a slight abuse of notations we keep the same notation $\hat U$) on the spectral parameter $k$ is not rational (because of $\lambda(k)$):
\[
\hat U(x,t,k) = \frac{\tilde{m}-1}{2}\left(\frac{1}{2\ii k}
            \begin{pmatrix}
                1 & 0 \\ 0 & -1
            \end{pmatrix}
            +\frac{\lambda(k)}{2\ii k}
            \begin{pmatrix}
                0 & 1 \\ -1 & 0
            \end{pmatrix}
            \right),
\]
which would complicate the construction of the RH problem, requiring either the introduction of a branch cut in the $k$ plane or the formulation of the RH problem on the Riemann sphere associated with $\lambda^2=4k^2+1$.

In order to avoid these complications, we introduce a new (uniformizing) spectral parameter $\mu$ such that both $\lambda$ and $k$ are rational w.r.t.~$\mu$:
\begin{equation}\label{la-k-mu}
\lambda=-\frac{1}{2}\left(\mu+\frac{1}{\mu}\right), \qquad
k = \frac{1}{4}\left(\mu-\frac{1}{\mu}\right).
\end{equation}
More precisely, we define $\mu=-\lambda-\ii\sqrt{1-\lambda^2}$, so that $k=-\frac{\ii}{2}\sqrt{1-\lambda^2}$ and $\sqrt{1-\lambda^2}=\frac{\ii}{2}\frac{\mu^2-1}{\mu}=2\ii k$. In terms of $\mu$ we have
\begin{align}\label{p_mu}
p(x,t,\mu)&=-\frac{\ii(\mu^2-1)}{4\mu}\left(\int_x^{+\infty} (\tilde m(\xi,t)-1)\dd\xi-x+\frac{8\mu^2}{(\mu^2+1)^2}t\right),\\
\label{U-hat_mu}
\hat U(x,t,\mu)&=
\frac{\ii(\mu^2+1)(\tilde m-1)}{2(\mu^2-1)}
\begin{pmatrix}
0 & 1 \\
-1 & 0 \\
\end{pmatrix}
-\frac{\ii\mu(\tilde m - 1)}{\mu^2-1}
\begin{pmatrix}
1 & 0 \\
0 & -1 \\
\end{pmatrix},
\end{align}
and, accordingly, equations \eqref{eq} become
\begin{equation}\label{inteq_mu}
\widetilde\Phi_{\pm}(x,t,\mu)=I+\int_{\pm\infty}^x\eul^{\frac{\ii(\mu^2-1)}{4\mu}\int_x^\xi\tilde m(\eta,t)\dd\eta\,\sigma_3}\hat U(\xi,t,\mu)\widetilde\Phi_{\pm}(\xi,t,\mu)\eul^{-\frac{\ii(\mu^2-1)}{4\mu}\int_x^\xi\tilde m(\eta,t)\dd\eta\,\sigma_3}\dd\xi.
\end{equation}

We are now able, by analogy with the case of the CH equation \cites{BS06,BS08}, to analyze the analytic and asymptotic properties of the solutions $\widetilde\Phi_\pm$ of \eqref{inteq_mu} as functions of $\mu$, using Neumann series expansions. Let $A^{(1)}$ and $A^{(2)}$ denote the columns of a $2\times 2$ matrix $A=\left(A^{(1)}\ \ A^{(2)}\right)$. Using these notations we have the following properties:
\begin{enumerate}[\textbullet]
\item
$\widetilde\Phi_-^{(1)}$ and $\widetilde\Phi_+^{(2)}$ are analytic in ${\D{C}}^+=\{\mu\in\D{C}\mid\Im \mu>0\}$;
\item
$\widetilde\Phi_+^{(1)}$ and $\widetilde\Phi_-^{(2)}$ are analytic in ${\D{C}}^-=\{\mu\in\D{C}\mid\Im\mu <0\}$;
\item
$\widetilde\Phi_-^{(1)}$, $\widetilde\Phi_+^{(2)}$, $\widetilde\Phi_+^{(1)}$, and $\widetilde\Phi_-^{(2)}$ are continuous up to the real line except at $\mu=\pm 1$.
\end{enumerate}
Further, we first observe that $\hat U(\mu)\equiv\hat U(x,t,\mu)$, $\hat V(\mu)\equiv\hat V(x,t,\mu)$ satisfy the same symmetries:
\begin{subequations}\label{sym-UV}
\begin{alignat}{3}\label{sym-U}
\hat{U}(\bar\mu)&=\sigma_1\overline{\hat{U}(\mu)}\sigma_1,&\qquad&\hat{U}(-\mu)=\sigma_2\hat{U}(\mu)\sigma_2,&\qquad&\hat{U}(\mu^{-1})=\sigma_1\hat U(\mu)\sigma_1,\\
\hat{V}(\bar\mu)&=\sigma_1\overline{\hat{V}(\mu)}\sigma_1,&&\hat{V}(-\mu)=\sigma_2\hat{V}(\mu)\sigma_2,&&\hat{V}(\mu^{-1})=\sigma_1\hat V(\mu)\sigma_1,
\end{alignat}
\end{subequations}
with $\mu\neq\pm1$, and also $\mu\neq 0$ for the symmetry $\mu\mapsto\mu^{-1}$. Moreover, $p(\mu)\equiv p(x,t,\mu)$ satisfies the following symmetries:
\begin{equation}\label{sym-p}
p^*(\mu)=-p(\mu)=p(-\mu)=p(\mu^{-1}).
\end{equation}
It follows that
\begin{enumerate}[\textbullet]
\item 
$\widetilde\Phi_\pm$ also satisfy the same symmetries as in \eqref{sym-U}:
\begin{equation}\label{sym-Phi}
\widetilde\Phi_\pm(\bar\mu)=\sigma_1\overline{\widetilde\Phi_\pm(\mu)}\sigma_1,\quad\widetilde\Phi_\pm(-\mu)=\sigma_2\widetilde\Phi_\pm(\mu)\sigma_2,\quad\widetilde\Phi_\pm(\mu^{-1})=\sigma_1\widetilde\Phi_\pm(\mu)\sigma_1.
\end{equation}
That means $\widetilde\Phi_\pm^{(1)}(\mu)=\sigma_1\widetilde\Phi_\pm^{(2)*}(\mu)=\sigma_3\sigma_1\widetilde\Phi_\pm^{(2)}(-\mu)=\sigma_1\widetilde\Phi_\pm^{(2)}(\mu^{-1})$ for $\pm\Im\mu\leq 0$, $\mu\neq\pm1$.
\end{enumerate}
In \eqref{comsys} the coefficients are traceless matrices, from which it follows that
\begin{enumerate}[\textbullet]
\item
$\det\widetilde\Phi_\pm\equiv 1$.
\end{enumerate}
Regarding the values of $\widetilde\Phi_\pm$ at particular points in the $\mu$-plane, \eqref{inteq_mu} implies the following:
\begin{enumerate}[\textbullet]
\item
$\left(\begin{smallmatrix}
\widetilde\Phi_-^{(1)} &
\widetilde\Phi_+^{(2)}\end{smallmatrix}\right)\to I$ as $\mu\to\infty$
with $\Im\mu\geq 0$, and also for $\mu=0$ (by the symmetry \eqref{sym-Phi}).
\item
$\left(\begin{smallmatrix}
\widetilde\Phi_+^{(1)} &
\widetilde\Phi_-^{(2)}\end{smallmatrix}\right)\to I$ as $\mu\to\infty$ with $\Im\mu\leq 0$, and also for $\mu=0$.
\item
As $\mu\to 1$, $\widetilde\Phi_\pm(x,t,\mu)=\frac{\ii}{2(\mu-1)}\alpha_\pm(x,t)\left(\begin{smallmatrix}-1&1\\ -1&1\end{smallmatrix}\right)+\ord(1)$ with $\alpha_\pm(x,t)\in\D{R}$ (understood column-wise, in the corresponding half-planes).
\item
As $\mu\to -1$, $\widetilde\Phi_\pm(x,t,\mu)=-\frac{\ii}{2(\mu+1)}\alpha_\pm(x,t)\left(\begin{smallmatrix}1 & 1 \\ -1 & -1\end{smallmatrix}\right)+\ord(1)$ with the same $\alpha_\pm(x,t)$ as the previous ones (by symmetry \eqref{sym-Phi}).
\end{enumerate}

\subsection{Spectral data}  \label{sec:spectral-data}

Introduce the scattering matrix $s(\mu)$ as a matrix relating $\widetilde\Phi_+$ and $\widetilde\Phi_-$ on the real line:
\begin{equation}\label{scat}
\widetilde\Phi_+(x,t,\mu)=\widetilde\Phi_-(x,t,\mu)
\eul^{-p(x,t,\mu)\sigma_3} s(\mu)\eul^{p(x,t,\mu)\sigma_3},\qquad\mu\in\D{R},\ \mu\neq\pm 1.
\end{equation}
By \eqref{sym-Phi}, $s(\mu)$ can be written in terms of two scalar spectral functions, $a(\mu)$ and $b(\mu)$:
\begin{equation}\label{scat_mat}
s(\mu)=\begin{pmatrix}\overline{a(\mu)} & b(\mu) \\ \overline{b(\mu)} & a(\mu)\end{pmatrix},\qquad\mu\in\D{R},
\end{equation}
satisfying the symmetries $\overline{a(\mu)}=a(-\mu)=a(\mu^{-1})$ and $\overline{b(\mu)}=-b(-\mu)=b(\mu^{-1})$ for $\mu\in\D{R}$.

The spectral functions $a(\mu)$ and $b(\mu)$ are uniquely determined by $u(x,0)$ through the solutions $\widetilde\Phi_\pm (x,0,\mu)$ of equations \eqref{inteq_mu}. On the other hand, using the representations \[
a(\mu)=\det\left(\widetilde\Phi_-^{(1)}\ \ \widetilde\Phi_+^{(2)}\right),\quad b(\mu)=\eul^{2p}\det\left(\widetilde\Phi_+^{(2)}\ \ \widetilde\Phi_-^{(2)}\right),
\]
the analytic properties of $\widetilde\Phi_\pm$ stated above imply corresponding properties of $a(\mu)$ and $b(\mu)$:
\begin{enumerate}[\textbullet]
\item
$a(\mu)$ can be analytically continued into $\D{C}^+$, being continuous up to the real line, except at $\mu=\pm1$. Moreover, $a(0)=1$, $a(\mu)\to 1$ as $\mu\to\infty$, and $a(\mu)$ satisfies the symmetries
\[
a(\mu)=\overline{a(-\bar\mu)}=a(-\mu^{-1})\text{ for }\Im\mu\geq 0.
\]
\item
$b(\mu)$ is continuous for $\mu\in\D{R}\setminus\{-1,1\}$. Moreover, $b(0)=0$ and $b(\mu)\to 0$ as $\mu\to\pm\infty$.
\item
As $\mu\to 1$, $a(\mu)=\gamma\frac{\ii}{2(\mu-1)}+\ord(1)$ and $b(\mu)=\gamma\frac{\ii}{2(\mu-1)}+\ord(1)$ with the same $\gamma\in\D{R}$, as follows from \eqref{scat}.
\item
As $\mu\to -1$, $a(\mu)=\gamma\frac{\ii}{2(\mu+1)}+\ord(1)$ and
$b(\mu)=-\gamma\frac{\ii}{2(\mu+1)}+\ord(1)$ with the same $\gamma$ as the previous one, by symmetry.
\item
$\abs{a(\mu)}^2-\abs{b(\mu)}^2=1$ for $\mu\in\D{R}$, $\mu\neq\pm1$.
\end{enumerate}

\begin{remark}\label{rem:sing}
The case $\gamma \neq 0$ is generic. On the other hand, in the non-generic case $\gamma =0$, we then have $a(\pm 1)=a_1$ and $b(\pm 1)=\pm b_1$ with some $a_1\in\D{R}$ and $b_1\in\D{R}$ such that $a_1^2=1+b_1^2$. It then follows from \eqref{scat} that the coefficients $\alpha_+(x,t)$ and $\alpha_-(x,t)$ appearing in the expansions of $\widetilde\Phi$ at $\mu=\pm1$ are related by
\begin{equation}\label{a1}
\alpha_+(x,t)=(a_1-b_1)\alpha_-(x,t).
\end{equation}
\end{remark}

\section{Riemann--Hilbert problem}\label{sec:3}
\subsection{RH problem parametrized by $\BS{(x,t)}$}

The analytic properties of $\widetilde\Phi_\pm$ stated above allow rewriting the scattering relation \eqref{scat} as a jump relation for a piece-wise meromorphic (w.r.t.~$\mu$), $2\times 2$-matrix valued function (depending on $x$ and $t$ as parameters). Indeed, define $M\equiv M(x,t,\mu)$ by
\begin{equation}\label{M}
M(x,t,\mu)=
\begin{cases}
\begin{pmatrix}\frac{\widetilde\Phi_-^{(1)}(x,t,\mu)}{a(\mu)} & \widetilde\Phi_+^{(2)}(x,t,\mu)\end{pmatrix},&\Im\mu>0,\\[1mm]
\begin{pmatrix}\widetilde\Phi_+^{(1)}(x,t,\mu) & \frac{\widetilde\Phi_-^{(2)}(x,t,\mu)}{\overline{a(\bar\mu)}}\end{pmatrix},&\Im\mu<0.
\end{cases}
\end{equation}
Define also
\begin{equation}\label{reflec}
r(\mu)\coloneqq\frac{b(\mu)}{a^*(\mu)},\quad\mu\in\D{R}.
\end{equation}
Then the limiting values $M_\pm(x,t,\mu)$, $\mu\in\D{R}$ of $M$ as $\mu$ is approached from $\D{C}^\pm$ are related by
\begin{subequations}\label{RH-x-J0}
\begin{equation}\label{RH-x}
M_-(x,t,\mu)=M_+(x,t,\mu)J(x,t,\mu),\quad \mu\in\D{R},\ \mu\neq \pm 1,
\end{equation}
where
\begin{equation}\label{J}
J(x,t,\mu)=\eul^{-p(x,t,\mu)\sigma_3}J_0(\mu)\eul^{p(x,t,\mu)\sigma_3}
\end{equation}
with
\begin{equation}\label{J0}
J_0(\mu)=\begin{pmatrix}
1&-r(\mu)\\r^*(\mu)&1-r(\mu)r^*(\mu)
\end{pmatrix}.
\end{equation}
\end{subequations}
Taking into account the properties of $\widetilde\Phi_\pm$ and $s(\mu)$ we check that $M(x,t,\mu)$ satisfies the following conditions:
\begin{enumerate}[\textbullet]
\item
The \emph{jump} condition \eqref{RH-x-J0} across $\D{R}$.
\item
The \emph{determinant} condition $\det M\equiv 1$.
\item 
The \emph{normalization} condition:
\begin{equation}\label{norm}
M\to I\quad\text{as }\mu\to\infty
\end{equation}
(and also $M(0)=I$ by symmetry, see \eqref{sym-M}).
\item 
\emph{Singularity} conditions:
\begin{equation}\label{sing}
M(x,t,\mu)=\begin{cases}
\frac{\ii\alpha_+(x,t)}{2(\mu-1)}\begin{pmatrix} -c & 1 \\ -c & 1 \end{pmatrix}+\ord(1), &\mu\to 1,\ \ \Im\mu>0,\\
-\frac{\ii\alpha_+(x,t)}{2(\mu +1)}\begin{pmatrix} c & 1 \\ -c & -1 \end{pmatrix}+\ord(1), &\mu\to -1,\ \ \Im\mu>0,
\end{cases}
\end{equation}
with some $\alpha_+(x,t)\in\D{R}$ and (see Remark~\ref{rem:sing})
\begin{subequations}\label{c}
\begin{equation}\label{c-1}
c\coloneqq\begin{cases}
0,&\text{if }\gamma\neq 0,\\
\frac{a_1+b_1}{a_1},&\text{if }\gamma=0,
\end{cases}
\end{equation}
where $a_1=a(1)$, $b_1=b(1)$, and $\gamma\coloneqq -2\ii\lim\limits_{\mu\to 1}(\mu-1)a(\mu)$. Notice that in terms of $r(\pm 1)$, the generic case $\gamma\neq 0$ corresponds to $r(1)=-r(-1)=-1$ whereas in the non-generic case, $\abs{r(\pm 1)}<1$ (see the case of the one-dimensional Schr\"odinger operator \cite{DT79}, which constitutes the spectral problem for the Korteweg--de Vries equation). Therefore, \eqref{c-1} can be written as 
\begin{equation}\label{c-2}
c\coloneqq\begin{cases}
0,&\text{if }r(1)=-1,\\
1+r(1)=1-r(-1),&\text{if }\abs{r(1)}<1.
\end{cases}
\end{equation}
\end{subequations}
Both conditions in \eqref{sing} are actually equivalent by the symmetries \eqref{sym-M}. 
\item 
\emph{Symmetries} (which result from \eqref{sym-Phi}):
\begin{equation}\label{sym-M}
M(\bar\mu)=\sigma_1\overline{M(\mu)}\sigma_1,\qquad M(-\mu)=\sigma_2M(\mu)\sigma_2,\qquad M(\mu^{-1})=\sigma_1M(\mu)\sigma_1,
\end{equation}
where $M(\mu)\equiv M(x,t,\mu)$. The first symmetry can also be written as $\sigma_1M^{(1)*}=M^{(2)}$. Moreover, \eqref{sym-M} implies the symmetries $\overline{M(-\bar\mu)}=M(-\mu^{-1})=\sigma_3M(\mu)\sigma_3$.
\end{enumerate}
If $a(\mu)$ is allowed to have zeros in $\D{C}^+$, the above conditions must be supplemented by residue conditions at these zeros. Assume that $a(\mu)$ has a finite number of simple zeros $\accol{\mu_j}_1^N$ in $\D{C}^+$. Symmetries $a(\mu)=\overline{a(-\bar\mu)}=a(-\mu^{-1})$ imply that this set of zeros is invariant under the transformations $\mu\mapsto-\bar\mu$ and $\mu\mapsto-\mu^{-1}$: for each $j$ there exist $j'$ and $j''$ such that $-\bar\mu_j=\mu_{j'}$ and $-\mu_j^{-1}=\mu_{j''}$.
\begin{enumerate}[\textbullet]
\item
\emph{Residue} conditions: $M^{(1)}(x,t,\mu)$ has simple poles at $\accol{\mu_j}_1^N$ and $M^{(2)}(x,t,\mu)$ has simple poles at $\accol{\bar\mu_j}_1^N$. Moreover
\begin{subequations}\label{res-M}
\begin{align}\label{res-M+}
\Res_{\mu_j}M^{(1)}(x,t,\mu)&=\frac{1}{\varkappa_j(x,t)}M^{(2)}(x,t,\mu_j),\\
\label{res-M-}
\Res_{\bar\mu_j}M^{(2)}(x,t,\mu)&=\frac{1}{\overline{\varkappa_j}(x,t)}M^{(1)}(x,t,\bar\mu_j).
\end{align}
\end{subequations}
Here $\varkappa_j(x,t)=\dot a(\mu_j)\delta_j\eul^{-2p(x,t,\mu_j)}$ with some constants $\delta_j\neq 0$. By symmetries \eqref{sym-M} both conditions in \eqref{res-M} are equivalent. Note also how the residue changes under the transformations $\mu\mapsto-\bar\mu$ and $\mu\mapsto-\mu^{-1}$: if $-\bar\mu_j=\mu_{j'}$ and $-\mu_j^{-1}=\mu_{j''}$ then $\varkappa_j=\overline{\varkappa_{j'}}=-\mu_j^{-2}\varkappa_{j''}$.
\end{enumerate}

\begin{proof}[Proof of \eqref{res-M}]
Indeed, let $\mu_j$ be a simple root of $a(\mu)$, that is, $a(\mu_j)=0$ with $\dot{a}(\mu_j)\neq 0$. Then, using $a(\mu)=\det\left(\widetilde\Phi_-^{(1)}\ \ \widetilde\Phi_+^{(2)}\right)=\det\left(\hat\Phi_-^{(1)}\ \ \hat\Phi_+^{(2)}\right)$, we have
\begin{subequations}\label{delta-j}
\begin{align}\label{delta-j-hat}
\hat\Phi_+^{(2)}(x,t,\mu_j)&=\delta_j\hat\Phi_-^{(1)}(x,t,\mu_j),\\
\label{delta-j-tilde}
\widetilde\Phi_+^{(2)}(x,t,\mu_j)&=\delta_j\eul^{-2p(x,t,\mu_j)}\widetilde\Phi_-^{(1)}(x,t,\mu_j)
\end{align}
\end{subequations}
with some constant $\delta_j\neq 0$. Hence, 
\[
\Res_{\mu_j}M^{(1)}(x,t,\mu)=\Res_{\mu_j}\frac{\widetilde\Phi_-^{(1)}(x,t,\mu)}{a(\mu)}=\frac{\widetilde\Phi_-^{(1)}(x,t,\mu_j)}{\dot{a}(\mu_j)}=\frac{\widetilde\Phi_+^{(2)}(x,t,\mu_j)}{\dot{a}(\mu_j)\delta_j\eul^{-2p(x,t,\mu_j)}}.
\]
Denoting $\varkappa_j(x,t)\coloneqq\dot a(\mu_j)\delta_j\eul^{-2p(x,t,\mu_j)}$ we get \eqref{res-M+}. The residue relation \eqref{res-M-} then follows by the symmetry $\mu\mapsto\mu^*=\bar\mu$. Indeed, applying this symmetry to \eqref{res-M+} and multiplying by $\sigma_1$ we get
\[
\Res_{\bar\mu_j}\sigma_1M^{(1)*}(x,t,\mu)=\frac{1}{\overline{\varkappa_j}(x,t)}\sigma_1M^{(2)*}(x,t,\bar\mu_j),
\]
which reduces to \eqref{res-M-} in view of the relation $\sigma_1M^{(1)*}=M^{(2)}$ (see \eqref{sym-M}).
\end{proof}

In the framework of the Riemann--Hilbert approach to nonlinear evolution equations, we  interpret the jump relation \eqref{RH-x}, normalization condition \eqref{norm}, singularity conditions \eqref{sing}, and residue conditions \eqref{res-M} as a Riemann--Hilbert problem, with the jump matrix and residue parameters determined by the initial data for the nonlinear problem. We proceed as in the case of the Camassa--Holm equation:
\begin{enumerate}[1)]
\item
In order to have the data for the RH problem to depend explicitly on the parameters, we use the space variable $y(x,t)\coloneqq x-\int_x^{+\infty}(\tilde m(\xi,t)-1)\dd\xi$ we have introduced in \eqref{shkala}.
\item
In order to determine an efficient way for retrieving the solution of the mCH equation from the solution of the RH problem, we pay a special attention to the behavior of the Jost solutions of the Lax pair equations at $\mu=\pm\ii$, i.e., at those values of $\mu$ that correspond to $\lambda=0$, when the $x$-equation \eqref{Lax-x}, \eqref{Lax-U} of the Lax pair becomes trivial (independent of the solution of the nonlinear equation in question).
\end{enumerate}
\subsection{Eigenfunction near $\BS{\mu=\ii}$}\label{sec:4}

In the case of the Camassa--Holm equation \cite{BS08} as well as other CH-type nonlinear integrable equations studied so far, see, e.g., \cites{BS13,BS15}, the analysis of the behavior of the respective Jost solutions at dedicated points in the complex plane of the spectral parameter (see Item 2) above) requires a dedicated gauge transformation of the Lax pair equations.

It is remarkable that in the case of the mCH equation, we don't need to use such a transformation; all we need is to regroup the terms in the Lax pair \eqref{Lax-1-x}, \eqref{Lax-1-t}.

Namely, let us rewrite \eqref{Lax-1-x} in terms of $\mu$ (keeping the same notation $\hat\Phi$ for the solution):
\begin{subequations}\label{Lax-2}
\begin{equation}\label{Lax-2-x}
\hat\Phi_x+\frac{\ii(\mu^2-1)}{4\mu}\sigma_3\hat\Phi = \hat U_0 \hat\Phi,
\end{equation}
where
\begin{equation}\label{U0-hat}
\hat U_0(x,t,\mu)\coloneqq\frac{\ii(\mu^2+1)(\tilde m - 1)}{2 (\mu^2-1)}
\begin{pmatrix}
0 & 1 \\
-1 & 0 \\
\end{pmatrix}
-\left(\frac{\ii\mu(\tilde m-1)}{\mu^2-1}+\frac{\ii(\mu^2-1)\tilde m}{4\mu}-\frac{\ii(\mu^2-1)}{4\mu}\right)\sigma_3,
\end{equation}
so that $\hat U_0(x,t,\pm\ii)\equiv 0$. Accordingly, rewrite \eqref{Lax-1-t} as
\begin{equation}\label{Lax-2-t}
\hat\Phi_t-\frac{2\ii(\mu^2-1)\mu}{(\mu^2+1)^2}\sigma_3\hat\Phi=\hat V_0 \hat\Phi,
\end{equation}
where
\begin{equation}\label{hat-V0}
\hat V_0(x,t,\mu)\coloneqq\frac{\ii(\mu^2-1)}{4\mu}(\tilde{u}^2-\tilde{u}_x^2+2\tilde{u})\tilde{m}\sigma_3 +\hat V(x,t,\mu).
\end{equation}
\end{subequations}
Further, introduce (compare with \eqref{p_mu})
\begin{equation}\label{p_0mu}
p_0(x,t,\mu)\coloneqq\frac{\ii(\mu^2-1)}{4\mu}x-\frac{2\ii(\mu^2-1)\mu}{(\mu^2+1)^2}t,
\end{equation}
then $Q_0\coloneqq p_0\sigma_3$, and $\widetilde\Phi_0\coloneqq\hat\Phi \eul^{Q_0}$ so that equations \eqref{Lax-2-x} and \eqref{Lax-2-t} become \begin{equation}\label{comsys-0}
\begin{cases}
\widetilde\Phi_{0x}+[Q_{0x},\widetilde\Phi_0]=\hat U_0\widetilde\Phi_0,&\\
\widetilde\Phi_{0t}+[Q_{0t},\widetilde\Phi_0]=\hat V_0\widetilde\Phi_0.&
\end{cases}
\end{equation}
Define the Jost solutions $\widetilde\Phi_{0\pm}$ of \eqref{comsys-0} as the solutions of the integral equations
\begin{equation}\label{inteq0_mu}
\widetilde\Phi_{0\pm}(x,t,\mu)=I+\int_{\pm\infty}^x\eul^{-\frac{\ii(\mu^2-1)}{4\mu}(x-\xi)\sigma_3}\hat U_0(\xi,t,\mu)\widetilde\Phi_{0\pm}(\xi,t,\mu)\eul^{\frac{\ii(\mu^2-1)}{4\mu}(x-\xi)\sigma_3}\dd\xi.
\end{equation}
If $\hat\Phi_{0\pm}\coloneqq\widetilde\Phi_{0\pm}\eul^{-p_0 \sigma_3}$ we observe that $\hat\Phi_{0\pm}(x,t,\mu)$ and $\hat\Phi_\pm(x,t,\mu)$ satisfy the same differential equations \eqref{Lax-2} and thus they are related by matrices $C_\pm(\mu)$ independent of $x$ and $t$:
\[
\hat\Phi_{\pm}=\hat\Phi_{0\pm}C_\pm(\mu).
\]
It follows that
\begin{equation}\label{inf_i_rel}
\widetilde\Phi_{\pm}(x,t,\mu)=\widetilde\Phi_{0\pm}(x,t,\mu)\eul^{-p_0(x,t,\mu) \sigma_3}C_\pm(\mu)\eul^{p(x,t,\mu)\sigma_3}.
\end{equation}
Thus, $C_\pm(\mu)=\eul^{(p_0(\pm\infty,t,\mu)-p(\pm\infty,t,\mu))\sigma_3}$. Since $p(x,t,\mu)-p_0(x,t,\mu)=-\frac{\ii(\mu^2-1)}{4\mu}\int_x^{+\infty}(\tilde m(\xi,t)-1)\dd\xi$ we find that $C_+(\mu)\equiv I$ whereas $C_-(\mu)=\eul^{\frac{\ii(\mu^2-1)}{4\mu}\int_{-\infty}^{+\infty}(\tilde m(\xi,t))-1)\dd\xi\,\sigma_3}$.

Since $\hat U_0(x,t,\ii)\equiv 0$, it follows from \eqref{inteq0_mu} that $\widetilde\Phi_{0\pm}(x,t,\ii)\equiv I$ and thus
\[
\widetilde\Phi_+(x,t,\ii)=\eul^{\frac{1}{2}\int_x^{+\infty}(\tilde m(\xi,t)-1)\dd\xi\,\sigma_3}\text{ and }\widetilde\Phi_{-}(x,t,\ii)=\eul^{-\frac{1}{2}\int_{-\infty}^x(\tilde m(\xi,t)-1)\dd\xi\,\sigma_3}.
\] 
Consequently,
\[
a(\ii)=\eul^{-\frac{1}{2}\int_{-\infty}^{+\infty}(\tilde m(\xi,t)-1)\dd\xi}
\]
and
\begin{subequations}\label{M(i)-2}
\begin{equation}\label{M(i)}
M(x,t,\ii) = \begin{pmatrix}
\eul^{\frac{1}{2}\int_x^{+\infty}(\tilde m(\xi,t)-1)\dd\xi} & 0 \\
0 & \eul^{-\frac{1}{2}\int_x^{+\infty}(\tilde m(\xi,t)-1)\dd\xi}
\end{pmatrix}.
\end{equation}
Then, by symmetry,
\begin{equation}\label{M(-i)}
M(x,t,-\ii) = \begin{pmatrix}
             \eul^{-\frac{1}{2}\int_x^{+\infty}(\tilde m(\xi,t)-1)\dd\xi} & 0 \\
             0 & \eul^{\frac{1}{2}\int_x^{+\infty}(\tilde m(\xi,t)-1)\dd\xi}
           \end{pmatrix}.
           \end{equation}
\end{subequations}

\begin{remark}\label{rem3-1}
The symmetries \eqref{sym-M} imply that $\overline{M(\ii)}=M(\ii)=\sigma_3M(\ii)\sigma_3$ where $M(\ii)\equiv M(x,t,\ii)$, and thus $M(\ii)$ is a diagonal matrix with real entries which, due to the determinant equality $\det M\equiv 1$, has the form
\begin{subequations}\label{M(pmi)}
\begin{equation}\label{M(pmi-a)}
M(x,t,\ii) = 
\begin{pmatrix}
\varphi(x,t) & 0 \\
0 & \varphi^{-1}(x,t)
\end{pmatrix}
\end{equation}
with some $\varphi(x,t)\in\D{R}$. Then, referring again to \eqref{sym-M}, it follows that
\begin{equation}\label{M(pmi-b)}					
M(x,t,-\ii)=\begin{pmatrix}
\varphi^{-1}(x,t) & 0 \\
0 & \varphi(x,t)
\end{pmatrix}
\end{equation}
\end{subequations}
with the same $\varphi(x,t)$. Therefore, the matrix structure of $M(x,t,\pm\ii)$ as in \eqref{M(i)-2} follows from the general properties of the solution of a Riemann--Hilbert problem (specified by jump, normalization, residue, singularity, and symmetry conditions). This is in contrast with the case of the Camassa--Holm equation \cites{BS06,BS08}, where a specific matrix structure of the solution of the associated RH problem, evaluated at a dedicated point ($k=\frac{\ii}{2}$ for the CH equation), constitutes an additional requirement for the solution. In that case, the proof of the uniqueness of the solution of the RH problem relies essentially on this additional property.
\end{remark}

In what follows we will use \eqref{M(i)-2} in order to extract the solution of the mCH equation from the solution of the associated RH problem.

\subsection{RH problem in the $\BS{(y,t)}$ scale}\label{sec:5}

Introducing the new space variable $y(x,t)$ by \eqref{shkala}, $\hat M(y,t,\mu)$ so that $M(x,t,\mu)=\hat M (y(x,t),t,\mu)$, the jump condition \eqref{RH-x} becomes
\begin{subequations} \label{Jp-y}
\begin{equation}\label{jump-y}
\hat M_-(y,t,\mu)=\hat M_+(y,t,\mu)\hat J(y,t,\mu),\qquad\mu\in\D{R},\quad\mu\neq\pm 1,
\end{equation}
where
\begin{equation}\label{J-J0}
\hat J(y,t,\mu)\coloneqq\eul^{-\hat p(y,t,\mu)\sigma_3}J_0(\mu)\eul^{\hat p(y,t,\mu)\sigma_3}
\end{equation}
with $J_0(\mu)$ defined by \eqref{J0} and
\begin{equation}\label{p-y}
\hat p(y,t,\mu)\coloneqq -\frac{\ii(\mu^2-1)}{4\mu}\left(-y+\frac{8\mu^2}{(\mu^2+1)^2}t\right).
\end{equation}
\end{subequations}
so that $J(x,t,\mu)=\hat J(y(x,t),t,\mu)$ and $p(x,t,\mu)=\hat p(y(x,t),t,\mu)$, where the jump $J(x,t,\mu)$ and the phase $p(x,t,\mu)$ are defined in \eqref{J} and \eqref{p_mu}, respectively.

Accordingly, in this scale, the residue conditions \eqref{res-M} become explicit as well:
\begin{equation}\label{res-M-hat}
\begin{split}
\Res_{\mu_j}\hat M^{(1)}(y,t,\mu)&=\frac{1}{\hat\varkappa_j(y,t)}\hat M^{(2)}(y,t,\mu_j),\\
\Res_{\bar\mu_j}\hat M^{(2)}(y,t,\mu)&=\frac{1}{\overline{\hat\varkappa_j}(y,t)}\hat M^{(1)}(y,t,\overline{\mu_j}),
\end{split}
\end{equation}
with $\hat\varkappa_j(y,t)=\dot a(\mu_j)\delta_j\eul^{-2\hat p(y,t,\mu_j)}$. Further we denote $\rho_j\coloneqq\dot a(\mu_j)\delta_j$.

Noticing that the normalization condition \eqref{norm}, the symmetries \eqref{sym-M}, and the singularity conditions \eqref{sing} at $\mu=\pm 1$ hold when using the new scale $(y,t)$, we arrive at the basic RH problem.

\begin{rh-pb*}
Given $r(\mu)$ for $\mu\in\D{R}$, $c\in\D{R}$, and $\accol{\mu_j,\rho_j}_1^N$ a set of points $\mu_j\in\D{C}^+$ and complex numbers $\rho_j\neq 0$ invariant by $\mu\mapsto-\bar\mu$ and $\mu\mapsto-\mu^{-1}$ (that is, $-\overline{\mu_j}=\mu_{j'}$ and $-\mu_j^{-1}=\mu_{j''}$ with $\rho_j=\overline{\rho_{j'}}=-\mu_j^{-2}\rho_{j''}$), find a piece-wise (w.r.t.~$\D{R}$) meromorphic, $2\times 2$-matrix valued function $\hat M(y,t,\mu)$ satisfying the following conditions:
\begin{enumerate}[\textbullet]
\item
The jump condition \eqref{Jp-y} across $\D{R}$ (with $J_0(\mu)$ defined by \eqref{J0}).
\item
The residue conditions \eqref{res-M-hat} with $\hat\varkappa_j(y,t)=\rho_j\eul^{-2\hat p(y,t,\mu_j)}$.
\item
The normalization condition $\hat M(y,t,\mu)\to I$ as $\mu\to\infty$.
\item
The symmetries
\begin{equation}\label{sym-M-hat}
\hat M(\bar\mu)=\sigma_1\overline{\hat M(\mu)}\sigma_1,\quad\hat M(-\mu)=\sigma_2\hat M(\mu)\sigma_2,\quad\hat M(\mu^{-1})=\sigma_1\hat M(\mu)\sigma_1
\end{equation}
where $\hat M(\mu)\equiv\hat M(y,t,\mu)$. These symmetries imply that $\hat M(-\mu^{-1})=\sigma_3\hat M(\mu)\sigma_3=\overline{\hat M(-\bar\mu)}$. 
\item
The singularity conditions
\begin{subequations}\label{sing-M-hat}
\begin{alignat}{3}\label{sing-M-hat-a}
\hat M(y,t,\mu)&=\frac{\ii\hat\alpha_+(y,t)}{2(\mu -1)}\begin{pmatrix} -c & 1 \\ -c & 1 \end{pmatrix}+\ord(1)&\quad&\text{as }\mu\to 1,&\;&\Im \mu> 0,\\
\label{sing-M-hat-b}
\hat M(y,t,\mu)&=-\frac{\ii\hat\alpha_+(y,t)}{2(\mu+1)}\begin{pmatrix}c&1\\ -c & -1 \end{pmatrix}+\ord(1)&&\text{as }\mu\to -1,&&\Im\mu> 0,
\end{alignat}
\end{subequations}
where $\hat\alpha_+(y,t)\in\D{R}$ is not specified. These two singularity conditions are actually equivalent by symmetries \eqref{sym-M-hat}.
\end{enumerate}
\end{rh-pb*}

\begin{data*}
Specific data for this RH problem can be derived from initial data of the Cauchy problem \eqref{mCH1-ic} satisfying $u_0(x)\to 1$ as $x\to\pm\infty$. 
\begin{enumerate}[\textbullet]
\item
We first get $s(\mu)$ through \eqref{scat} at $t=0$ (using the solutions of \eqref{inteq_mu} taken at $t=0$). 
\item
Spectral data $a(\mu)$, $b(\mu)$, and $r(\mu)$ follow through \eqref{scat_mat} and \eqref{reflec}. 
\item
Then $\accol{\mu_j}_1^N$ are the zeros of $a(\mu)$ in $\D{C}^+$. 
\item
The real constant $c$ is defined through \eqref{c}. 
\item
The constants $\accol{\delta_j}_1^N$ are defined by \eqref{delta-j-tilde} at $t=0$ (using the solutions of \eqref{inteq} at $t=0$). 
\item
Finally, the $\accol{\rho_j}_1^N$ are defined by $\rho_j=\dot a(\mu_j)\delta_j$.
\end{enumerate}
\end{data*}

Further, the basic RH problem associated with the Cauchy problem \eqref{mCH1-ic} for the mCH equation is the basic RH problem with data associated with initial data satisfying $u_0(x)\to 1$, as we just specified.

\begin{remark}
An important difference between the cases of the CH and mCH equations is that in the former case, there is a possibility to reduce the matrix RH problems to vector ones which have no singularity at a point on the contour: this can be done by  multiplying the respective $\hat M$ by the vector $(1,1)$ from the left. This trick will obviously not work in our current case, since the matrix structure (see \eqref{sing-M-hat}) of the singularity at $\mu=1$ is different from that at $\mu=-1$.
\end{remark}

\subsection{Uniqueness of the solution of the basic RH problem}\label{sec:6}

Assume that the RH problem \eqref{Jp-y}--\eqref{sing-M-hat} has a solution $\hat M $. In order to prove that this solution is unique, we first observe that $\det\hat M\equiv 1$. 

Indeed, the conditions for $\hat M$ imply that $\det\hat M$ has neither a jump across $\D{R}$ no singularities at $\mu_j$. Moreover, $\det\hat M$ tends to $1$ as $\mu\to\infty$, and the only possible singularities of $\det\hat M$ are simple poles at $\mu=\pm 1$. Then, by Liouville's theorem, $\det\hat M \equiv 1 +\frac{\phi_1}{\mu-1}+\frac{\phi_2}{\mu+1}$ with some $\phi_j$. But then, the symmetry $\hat M(\mu^{-1})=\sigma_1\hat M(\mu)\sigma_1$ from \eqref{sym-M-hat} implies that $\phi_1 = \phi_2 =0$ and thus $\det\hat M \equiv 1$.

Now suppose that $\hat M_1$ and $\hat M_2$ are two solutions of the RH problem, and consider $P\coloneqq\hat M_1(\hat M_2)^{-1}$. Obviously, $P$ has neither a jump across $\D{R}$ no singularities at $\mu_j$. Moreover, $P$ tends to $I$ as $\mu\to\infty$, and the only possible singularities of $P$ are simple poles at $\mu=\pm 1$.

Consider, for example, the development of $\hat M_j$, $j=1,2$ as $\mu\to -1$ with $\Im\mu>0$:
\[
\hat M_j(y,t,\mu)=-\frac{\ii\beta_j(y,t)}{2(\mu+1)}\begin{pmatrix} c & 1 \\ -c & -1 \end{pmatrix}+\begin{pmatrix} n_j(y,t) & m_j(y,t) \\ f_j(y,t) & g_j(y,t) \end{pmatrix}+\ord(\mu+1)\text{ as }\mu\to -1,\ \mu \in\D{C}^+.
\]
By $\det\hat M_j\equiv 1$ it follows that
\[
(\hat M_j(y,t,\mu))^{-1}=-\frac{\ii\beta_j(y,t)}{2(\mu +1)}\begin{pmatrix} -1 & -1 \\ c & c \end{pmatrix}+\begin{pmatrix} g_j(y,t) & -m_j(y,t) \\ -f_j(y,t) & n_j(y,t)\end{pmatrix}+\ord(\mu+1)\text{ as }\mu\to -1,\ \mu \in\D{C}^+.
\]
Moreover, using these expressions to calculate the expansion of $\hat M_j\hat M_j^{-1}$ as $\mu\to-1$ the vanishing of the term of order $(\mu+1)^{-1}$ reads as
\begin{equation} \label{sing-rel}
n_j(y,t)+f_j(y,t)=c(m_j(y,t)+g_j(y,t)),\quad j=1,2.
\end{equation}
Hence, \eqref{sing-rel} implies that
\[
P(y,t,\mu)=-\frac{\ii\psi(y,t)}{2(\mu+1)}\begin{pmatrix}1&1\\-1&-1\end{pmatrix}+\ord(1)\text{ as }\mu\to-1,\ \mu\in\D{C}^+,
\]
for some $\psi(y,t)$. Then, by the symmetry $P(\mu^{-1})=\sigma_3P(\mu)\sigma_3$, we have
\[
P(y,t,\mu)=-\frac{\ii\psi(y,t)}{2(\mu-1)}\begin{pmatrix}1&-1\\1&-1\end{pmatrix}+\ord(1)\text{ as }\mu\to 1,\ \mu \in\D{C}^+,
\]
and, according to the Liouville theorem and the normalization condition,
\[
P=-\frac{\ii}{2}\psi(y,t)\left(\frac{1}{\mu-1}\begin{pmatrix}1&-1\\1&-1 \end{pmatrix}+\frac{1}{\mu+1}\begin{pmatrix}1&1\\-1&-1\end{pmatrix}\right)+I.
\]
Evaluating this at $\mu=\ii$ we have
\begin{equation}\label{N-i-1}
P(y,t,\ii)=-\frac{\ii}{2}\psi(y,t)
\begin{pmatrix}
-\ii & 1 \\
-1 &\ii
\end{pmatrix}+I.
\end{equation}
But, according to \eqref{M(pmi-a)}, both matrices $\hat M_1(\ii)$ and $\hat M_2(\ii)$ are diagonal. Hence $P(y,t,\ii)$ is also diagonal and \eqref{N-i-1} implies that $\psi(y,t)\equiv 0$. Consequently, $P(y,t,\mu)\equiv I$ so that $\hat M_1\equiv\hat M_2$.

\subsection{Recovering $\BS{u(x,t)}$ from the solution of the RH problem}\label{sec:recover}

We will show how to recover the solution of the Cauchy problem \eqref{mCH1-ic} from the solution of the basic RH problem whose data are associated with the initial data $u_0(x)$. We begin with some preliminary observations.

Going back to the construction of $M(x,t,\mu)$ from the Jost solutions, see Section \ref{sec:4}, we can use \eqref{M(i)} in order to express the solution $u(x,t)$ of the mCH equation in terms of $M(x,t,\mu)$ evaluated at $\mu=\ii$. Indeed, introduce (compare with the case of the CH equation \cite{BS08})
\begin{align*}
&\tilde\mu_1(x,t)\coloneqq M_{11}(x,t,\ii)+M_{21}(x,t,\ii)= \eul^{\frac{1}{2}\int_x^{+\infty}(\tilde m(\xi,t)-1)\dd\xi},\\
&\tilde\mu_2(x,t)\coloneqq M_{12}(x,t,\ii)+M_{22}(x,t,\ii)=\eul^{-\frac{1}{2}\int_x^{+\infty}(\tilde m(\xi,t)-1)\dd\xi}.
\end{align*}
Using the new space variable $y(x,t)\coloneqq x-\int_x^{+\infty}(\tilde m(\xi,t)-1)\dd\xi$ we have introduced in \eqref{shkala}, the above equations yield
\begin{equation}
\label{mu1_mu2}
\frac{\tilde\mu_1(x,t)}{\tilde\mu_2(x,t)}=\eul^{\int_x^{+\infty}(\tilde m(\xi,t)-1)\dd\xi}=\eul^{x-y(x,t)}
\end{equation}
and thus
\begin{equation}\label{x-y}
x=y(x,t)+\ln\frac{\tilde\mu_1 (x,t)}{\tilde\mu_2(x,t)}.
\end{equation}
Also notice that
\begin{equation}\label{mu1mu2}
\tilde\mu_1(x,t)\tilde\mu_2(x,t)= 1.
\end{equation}

\begin{proposition}\label{prop-recover}
Let $\hat M(y,t,\mu)$ be the solution of the RH problem \eqref{Jp-y}--\eqref{sing-M-hat} whose data are associated with the initial data $u_0(x)$. Define $\hat\mu_1(y,t)\coloneqq\hat M_{11}(y,t,\ii)+\hat M_{21}(y,t,\ii)$ and $\hat\mu_2(y,t)\coloneqq\hat M_{12}(y,t,\ii)+\hat M_{22}(y,t,\ii)$. The solution $u(x,t)$ of the Cauchy problem \eqref{mCH1-ic} has $x$-derivative given by the parametric representation
\begin{subequations}\label{recover}
\begin{align}\label{u_x(y,t)}
&u_x(x+t,t)=\frac{1}{2}\partial_t\ln\frac{\hat\mu_1(y,t)}{\hat\mu_2(y,t)},\\
\label{x(y,t)}
&x(y,t)=y+\ln\frac{\hat\mu_1(y,t)}{\hat\mu_2(y,t)}.
\end{align}
\end{subequations}
\end{proposition}

\begin{proof}
In what follows we will express $\tilde u_x$ in the variables $(y,t)$. To express a function $\tilde f(x,t)$ in $(y,t)$ we will use the notation $\hat f(y,t)\coloneqq\tilde f(x(y,t),t)$, e.g.,
\[
\hat u(y,t)\coloneqq\tilde u(x(y,t),t),\ \hat u_x(y,t)\coloneqq\tilde u_x(x(y,t),t),\ \hat m(y,t)\coloneqq\tilde m(x(y,t),t),\ \hat \omega(y,t)\coloneqq\tilde\omega(x(y,t),t).
\]
Differentiation of the identity $x(y(x,t),t)=x$ w.r.t.~$t$ gives
\begin{equation}\label{dt-0}
\partial_t\left(x(y(x,t),t)\right)=x_y(y,t)y_t(x,t)+x_t(y,t)=0.
\end{equation}
From \eqref{shkala} it follows that
\begin{equation}\label{x_y}
x_y(y,t)=\frac{1}{\hat m(y,t)}
\end{equation}
and $y_t(x,t)=-\int_x^{+\infty}\tilde m_t(\xi,t)\dd\xi$. By \eqref{mCH-2}, the latter equality becomes
\[
y_t(x,t)=\int_x^{+\infty}\left(\tilde\omega\tilde m\right)_\xi(\xi,t) \dd\xi=-\tilde\omega\tilde m(x,t).
\]
Substituting this and \eqref{x_y} into \eqref{dt-0} we obtain
\begin{equation}\label{x_t}
x_t(y,t)=\hat\omega(y,t).
\end{equation}
Further, differentiating \eqref{x_t} w.r.t.~$y$ we get
\begin{equation}\label{x_yt}
x_{ty}(y,t)=\hat\omega_xx_y(y,t)=2\hat u_x(\hat u-\hat u_{xx}+1)\frac{1}{\hat m}(y,t)=2\hat u_x(y,t).
\end{equation}
Therefore, we arrive at a parametric representation of $\tilde u_x(x,t)$:
\begin{align*}
&\tilde u_x(x(y,t),t)\equiv\hat u_x(y,t)=\frac{1}{2}\partial_tx(y,t),\\
&x(y,t)=y+\frac{\ln\hat\mu_1(y,t)}{\ln\hat\mu_2(y,t)},
\end{align*}
which yields \eqref{recover}. For the direct determination of $u$ from the solution of the RH problem, see Remark~\ref{rem-u} below.
\end{proof}

\begin{remark}
In the case of the Camassa--Holm equation, the relation between the new and original space variables \eqref{x-y} is the same whereas the derivative \eqref{x_t} gives directly the solution $u$ of the nonlinear equation (in the $(y,t)$ variables) in question.
\end{remark}

\section{From a solution of the RH problem to a solution of the mCH equation}\label{sec:7}

Henceforth we consider a RH problem \eqref{Jp-y}--\eqref{sing-M-hat} with data not necessarily related to initial data for the mCH equation. This section aims to show that starting from the solution $\hat M(y,t,\mu)$ of such a RH problem one can construct a solution (at least, locally) of the mCH equation by manipulations similar to those of Section \ref{sec:recover}. For this purpose, we will show that starting from $\hat M(y,t,\mu)$ one can define $2\times 2$-matrix valued functions $\hat\Psi(y,t,\mu)$ satisfying Lax pair equations
\begin{align*}
\hat\Psi_y  &= \doublehat{U}\hat\Psi, \\
\hat \Psi_t &= \doublehat{V}\hat\Psi,
\end{align*}
whose coefficients $\doublehat{U}$ and $\doublehat{V}$ are obtained from $\hat M(y,t,\mu)$, and whose compatibility condition is the mCH equation (written in the $(y,t)$ variables).

First, let us reformulate the original Lax pair equations \eqref{Lax-1} in the $(y,t)$ variables. Introducing $\hat\Psi(y,t) = \hat\Phi (x(y,t),t)$ and taking into account \eqref{x_t} and \eqref{x_y}, the Lax pair \eqref{Lax-1} in the variables $(y,t)$ takes the form:
\begin{align*}
\hat\Psi_y +\ii k\sigma_3 \hat\Psi &= \frac{\tilde{m}-1}{\tilde m}\frac{\lambda}{4\ii k}
\begin{pmatrix}
\frac{1}{\lambda} & 1 \\ -1 & -\frac{1}{\lambda}
\end{pmatrix}\hat\Psi,\\
\hat\Psi_t-\frac{2\ii k}{\lambda^2}\sigma_3\hat\Psi&=\left(\frac{\tilde u}{2\ii k}\begin{pmatrix}-1&-\frac{1}{\lambda}\\\frac{1}{\lambda}& 1\end{pmatrix}+\frac{\tilde u_x}{\lambda}\begin{pmatrix}0&1\\1&0\end{pmatrix}\right)\hat\Psi,
\end{align*}
where $k\coloneqq-\frac{\ii}{2}\sqrt{1-\lambda^2}$.

Consequently, using $\mu$ as spectral parameter (see \eqref{la-k-mu}), we have

\begin{proposition}
The Lax pair \eqref{Lax-1} has the following form in the variables $(y,t,\mu)$:
\begin{equation}\label{hatpsieq}
\begin{split}
&\hat\Psi_y+\frac{\ii(\mu^2-1)}{4\mu}\sigma_3\hat\Psi=\widetilde{U}\hat \Psi,\\
&\hat \Psi_t-\frac{2\ii(\mu^2-1)\mu}{(\mu^2+1)^2}\sigma_3\hat\Psi=\widetilde{V}\hat\Psi,
\end{split}
\end{equation}
where
\begin{subequations}\label{tilde-UV}
\begin{align}\label{tilde-U}
\widetilde U(y,t,\mu)&=\frac{\ii f(y,t)}{\mu-1}
\begin{pmatrix}
1 & -1 \\ 1 & -1
\end{pmatrix}+\frac{\ii f(y,t)}{\mu+1}
\begin{pmatrix}
1 & 1 \\ -1 & -1
\end{pmatrix}+\ii f(y,t)
\begin{pmatrix}
0 & -1 \\ 1 & 0
\end{pmatrix}, \\ 
\label{tilde-V}
\widetilde V(y,t,\mu)&=\frac{\ii q(y,t)}{\mu-1}
\begin{pmatrix}
1 & -1 \\ 1 & -1
\end{pmatrix}+\frac{\ii q(y,t)}{\mu+1}
\begin{pmatrix}
1 & 1 \\ -1 & -1
\end{pmatrix}\notag\\
&\quad
+\frac{1}{\mu-\ii}
\begin{pmatrix}
0 & g_1(y,t) \\ g_2(y,t) & 0
\end{pmatrix}+\frac{1}{\mu+\ii}
\begin{pmatrix}
0 & g_2(y,t) \\ g_1(y,t) & 0
\end{pmatrix},						
\end{align}
\end{subequations}
with $f$, $q$, $g_1$, and $g_2$ as follows:
\begin{equation}\label{f-q-g}
f=-\frac{\hat m-1}{2\hat m},\quad q=\hat u,\quad g_1 =-\hat u-\hat u_x,\quad g_2=\hat u-\hat u_x.
\end{equation}
\end{proposition}

Our goal in this section is to show that giving a solution $\hat M(y,t,\mu)$ to the RH problem \eqref{Jp-y}--\eqref{sing-M-hat}, where the data $r(\mu)$ for $\mu\in\D{R}$, $c\in\D{R}$, and $\accol{\mu_j,\rho_j}_1^N$ are not a priori associated with some initial data $u_0(x)$, one can ``extract'' from $\hat M(y,t,\mu)$ a solution to the mCH equation. The idea is as follows:
\begin{enumerate}[(a)]
\item 
Starting from $\hat M(y,t,\mu)$, define $\hat\Psi(y,t,\mu) = \hat M(y,t,\mu) \eul^{-\hat p(y,t,\mu) \sigma_3}$ and show that $\hat\Psi(y,t,\mu)$ satisfies the system of differential equations:
\begin{equation}\label{Lax-hat-hat}
\begin{split}
\hat\Psi_y&=\doublehat{U}\hat\Psi, \\
\hat \Psi_t&=\doublehat{V}\hat\Psi,
\end{split}
\end{equation}
where $\doublehat{U}$ and $\doublehat{V}$ have the same (rational) dependence on $\mu$ as in \eqref{hatpsieq} and \eqref{tilde-UV}, with coefficients given in terms of $\hat M(y,t,\mu)$ evaluated at appropriate values of $\mu$.
\item
Show that the compatibility condition for \eqref{Lax-hat-hat}, which is the equality $\doublehat{U}_t - \doublehat{V}_y + [\doublehat{U},\doublehat{V}]=0$, reduces to the mCH equation.
\end{enumerate}

\begin{proposition}\label{prop-y}
Let $\hat M(y,t,\mu)$ be the solution of the RH problem \eqref{Jp-y}--\eqref{sing-M-hat}. Define
\begin{equation}\label{hatpsiy}
\hat\Psi(y,t,\mu)\coloneqq\hat M(y,t,\mu)\eul^{-\hat p(y,t,\mu)\sigma_3},
\end{equation}
where $\hat p(y,t,\mu)\coloneqq-\frac{\ii(\mu^2-1)}{4\mu}\left(-y+\frac{8\mu^2}{(\mu^2+1)^2}t\right)$. Then $\hat\Psi(y,t,\mu)$ satisfies the differential equation
\[
\hat\Psi_y=\doublehat{U}\hat\Psi
\]
with $\doublehat{U}=-\frac{\ii(\mu^2-1)}{4\mu}\sigma_3+\widetilde{U}$, where $\widetilde{U}$ is as in \eqref{tilde-U} with $f$ given by 
\[
f(y,t)\coloneqq-\frac{\eta(y,t)}{2},
\] 
$\eta(y,t)$ being extracted from the large $\mu$ expansion of $\hat M(y,t,\mu)$:
\[
\hat M(y,t,\mu)=I+\frac{1}{\mu}
\begin{pmatrix}\xi(y,t) & \eta(y,t)\\\eta(y,t)&-\xi(y,t)\end{pmatrix}+\ord(\mu^{-2}),\qquad\mu\to\infty.
\]
\end{proposition}

\begin{proof}
First, notice that $\hat\Psi(y,t,\mu)$ satisfies the jump condition
\[
\hat\Psi_-(y,t,\mu)=\hat\Psi_+(y,t,\mu)J_0(\mu)
\]
with the jump matrix $J_0$ independent of $y$. Hence, $\hat\Psi_y(y,t,\mu)$ satisfies the same jump condition. Consequently, $\hat\Psi_y \hat\Psi^{-1}=\hat M_y\hat M^{-1}-\hat p_y\hat M\sigma_3\hat M^{-1}$ has no jump and thus it is a meromorphic function, with possible singularities at $\mu=\infty$, $\mu=0$, and $\mu=\pm 1$. Let us evaluate $\hat\Psi_y\hat\Psi^{-1}$ near these points.

(i) 
As $\mu\to\infty$, we have $\hat p_y=\frac{\ii\mu}{4}+\ord(\mu^{-1})$ and thus
\[
\hat\Psi_y \hat\Psi^{-1} = -\frac{\ii\mu}{4}\sigma_3-\frac{\ii}{4}[\hat M^{(\infty)}, \sigma_3] +\ord(\mu^{-1}),
\]
where $\hat M^{(\infty)}\equiv\hat M^{(\infty)}(y,t)$ comes from the large $\mu$ asymptotics of $\hat M$:
\[
\hat M=I+\frac{\hat M^{(\infty)}}{\mu}+\ord(\mu^{-2}),\qquad \mu\to\infty.
\]
Symmetries \eqref{sym-M-hat} imply that $\sigma_2\hat M^{(\infty)}\sigma_2=-\hat M^{(\infty)}$ and $\sigma_1\hat M^{(\infty)}\sigma_1=\overline{\hat M^{(\infty)}}$, so that
\[
\hat M^{(\infty)} = \begin{pmatrix}
\xi & \eta \\ \eta & -\xi
\end{pmatrix}
\]
with some $\xi(y,t)\in\ii\D{R}$ and $\eta(y,t)\in\D{R}$. Consequently,
\begin{equation}\label{psi-inf}
\hat\Psi_y\hat\Psi^{-1}=-\frac{\ii\mu}{4}\sigma_3-\frac{\ii}{2}\begin{pmatrix}0 & -\eta \\ \eta & 0\end{pmatrix}+\ord(\mu^{-1}),\qquad \mu\to\infty.
\end{equation}
Then, by symmetry,
\begin{equation}\label{psi-0}
\hat\Psi_y \hat\Psi^{-1} = \frac{\ii}{4\mu}\sigma_3 +\frac{\ii}{2}
\begin{pmatrix}
0 & -\eta \\ \eta & 0
\end{pmatrix}+\ord(\mu),\qquad\mu\to 0.
\end{equation}

(ii) 
Pushing the expansion \eqref{sing-M-hat-a} of $\hat M(\mu)$ a step further, and proceeding as in Section~\ref{sec:6} to get \eqref{sing-rel} we have
\begin{equation}\label{psi-1}
\hat\Psi_y\hat\Psi^{-1}=\frac{\ii\beta_1}{\mu-1}\begin{pmatrix}
                  1 & -1 \\
                  1 & -1\end{pmatrix}+\ord(1),\qquad \mu\to 1,
\end{equation}
with some $\beta_1(y,t)\in\D{R}$. By symmetry,
\begin{equation}\label{psi+1}
\hat\Psi_y\hat\Psi^{-1}=\frac{\ii\beta_1}{\mu+1}\begin{pmatrix}
                  1 & 1 \\
                  -1 & -1\end{pmatrix}+\ord(1), \qquad \mu\to -1.
\end{equation}

Combining \eqref{psi-inf}, \eqref{psi-0}, \eqref{psi-1}, and \eqref{psi+1}, we obtain that the function
\[
\hat\Psi_y \hat\Psi^{-1}+\frac{\ii(\mu^2-1)}{4\mu}\sigma_3
 - \frac{\ii\beta_1}{\mu-1}\begin{pmatrix}
                  1 & -1 \\
                  1 & -1
                \end{pmatrix}
- \frac{\ii\beta_1}{\mu+1}\begin{pmatrix}
                  1 & 1 \\
                  -1 & -1
                \end{pmatrix}
+\frac{\ii}{2}
\begin{pmatrix}
0 & -\eta \\ \eta & 0
\end{pmatrix}								
\]
is holomorphic in the whole complex $\mu$-plane and, moreover, vanishes as $\mu\to\infty$. Then, by Liouville's theorem, it vanishes identically.

Further, again by symmetry, $\hat M(y,t,\ii)$ is diagonal (see Remark \ref{rem3-1}), which implies that the following sum is diagonal as well:
\[
\frac{\ii\beta_1}{\ii-1}\begin{pmatrix}
                  1 & -1 \\
                  1 & -1
                \end{pmatrix}
+\frac{\ii\beta_1}{\ii+1}\begin{pmatrix}
                  1 & 1 \\
                  -1 & -1
                \end{pmatrix}
- \frac{\ii}{2}
\begin{pmatrix}
0 & -\eta \\ \eta & 0
\end{pmatrix}.
\]
It follows that $\frac{\eta}{2}=-\beta_1$, and thus we arrive at the equality $\hat\Psi_y=\doublehat{U}\hat\Psi$ with $\doublehat{U}=-\frac{\ii(\mu^2-1)}{4\mu}\sigma_3+\widetilde U$, where
$\widetilde U$ is as in \eqref{tilde-U} with $f=\beta_1$.
\end{proof}

\begin{proposition}\label{prop-t}
The function $\hat\Psi(y,t,\mu)$ defined by \eqref{hatpsiy} satisfies the differential equation
\begin{equation}\label{psi-t-i}
\hat\Psi_t=\doublehat{V}\hat\Psi
\end{equation}
with $\doublehat{V}=\frac{2\ii(\mu^2-1)\mu}{(\mu^2+1)^2}\sigma_3+\widetilde V$, where $\widetilde V$ is as in \eqref{tilde-V} with coefficients $q$, $g_1$, and $g_2$ determined by evaluating $\hat M(y,t,\mu)$ as $\mu\to 1$ and $\mu\to\ii$.
\end{proposition}

\begin{proof}
Similarly to Proposition \ref{prop-y}, we notice that $\hat\Psi_t \hat\Psi^{-1}=\hat M_t\hat M^{-1}-\hat p_t\hat M\sigma_3\hat M^{-1}$ has no jump and thus it is a meromorphic function, with possible singularities at $\mu=\infty$, $\mu=0$, $\mu=\pm 1$, and $\mu=\pm\ii$, the latter being due to the singularity of $\hat p_t$ at $\mu=\pm\ii$:
\begin{equation}\label{psi-t-i-1}
\hat p_t(\mu)=\pm\frac{1}{(\mu\mp\ii)^2}-\frac{\ii}{\mu\mp\ii}+\ord(1), \qquad\mu\to\pm\ii.
\end{equation}
Evaluating $\hat\Psi_t\hat\Psi^{-1}$ near these points, we have the following.
\begin{enumerate}[(i)]
\item
As $\mu\to\infty$, we have $\hat p_t(\mu)=\ord(\mu^{-1})$ and thus
\begin{equation}\label{psi-inf-t}
\hat\Psi_t\hat\Psi^{-1}(\mu)=\ord(\mu^{-1}),\qquad\mu\to\infty.
\end{equation}
Then, by symmetry,
\begin{equation}\label{psi-0-t}
\hat\Psi_t\hat\Psi^{-1}(\mu)=\ord(\mu),\qquad\mu\to 0.
\end{equation}
\item
Expanding $\hat M(\mu)$ at $\mu=1$, and proceeding as above to get \eqref{psi-1}, we have
\begin{equation}\label{psi-1-t}
\hat\Psi_t\hat\Psi^{-1}(\mu)=\frac{\ii\beta_2}{\mu-1}\begin{pmatrix}
                  1 & -1 \\
                  1 & -1
                \end{pmatrix}+\ord(1),\qquad \mu\to 1,
\end{equation}
with some $\beta_2(y,t)\in\D{R}$. By symmetry,
\begin{equation}\label{psi+1-t}
\hat\Psi_t\hat\Psi^{-1}(\mu)=\frac{\ii\beta_2}{\mu+1}\begin{pmatrix}
                  1 & 1 \\
                  -1 & -1
                \end{pmatrix}+\ord(1), \qquad \mu\to -1.
\end{equation}
\item
Evaluating $\hat M(\mu)$ as $\mu\to\ii$, we first notice that, due to symmetries,
\begin{equation}\label{hat-M-i}
\hat M(\mu)=\begin{pmatrix}
                  a_1 & 0 \\
                  0 & a_1^{-1}
\end{pmatrix}  
+\begin{pmatrix}
                  0 & a_2 \\
                  a_3 & 0
\end{pmatrix}(\mu-\ii)+\ord((\mu-\ii)^2), \qquad \mu\to\ii,
\end{equation}
with some $a_j\equiv a_j(y,t)$, $j=1,2,3$. Taking into account \eqref{psi-t-i-1}, we have
\begin{equation}\label{hathatVi}
\hat\Psi_t\hat\Psi^{-1}(\mu)=-\frac{1}{(\mu-\ii)^2}\sigma_3+\frac{1}{\mu-\ii}\left(\ii\sigma_3+\begin{pmatrix}
                  0 & 2a_2 a_1 \\
                 -2 a_3 a_1^{-1} & 0
\end{pmatrix}\right)+\ord(1),\qquad \mu\to\ii.
\end{equation}
Then, by symmetry,
\begin{equation}\label{hathatV-i}
\hat\Psi_t\hat\Psi^{-1}(\mu)=\frac{1}{(\mu+\ii)^2}\sigma_3+\frac{1}{\mu+\ii}\left(\ii\sigma_3 +\begin{pmatrix}
                  0 & -2 a_3 a_1^{-1}  \\
                  2a_2 a_1 & 0
\end{pmatrix}\right) +\ord(1),\qquad \mu\to -\ii.
\end{equation}
\end{enumerate}
Combining \eqref{psi-inf-t}, \eqref{psi-1-t}, and \eqref{psi+1-t}, \eqref{hathatVi}, and \eqref{hathatV-i},
we obtain that the function
\begin{align*}
\hat\Psi_t\hat\Psi^{-1}(\mu)-\frac{2\ii(\mu^2-1)\mu}{(\mu^2+1)^2}\sigma_3
&-\frac{1}{\mu-1}\ii\beta_2
\begin{pmatrix}
1 & -1 \\ 1 & -1
\end{pmatrix}-\frac{1}{\mu+1}\ii\beta_2
\begin{pmatrix}
1 & 1 \\ -1 & -1
\end{pmatrix}\\
&-\frac{1}{\mu-\ii}
\begin{pmatrix}
0 & \gamma_1 \\ \gamma_2 & 0
\end{pmatrix}- \frac{1}{\mu+\ii}
\begin{pmatrix}
0 & \gamma_2 \\ \gamma_1 & 0
\end{pmatrix}	
\end{align*}			
with $\gamma_1=2a_2a_1$ and $\gamma_2=-2a_3a_1^{-1}$
is holomorphic in the whole complex $\mu$-plane and, moreover, vanishes as $\mu\to\infty$. Then, by Liouville's theorem, it vanishes identically. Thus we arrive at the equality $\hat\Psi_t=\doublehat{V}\hat\Psi$ with $\doublehat{V}(\mu)=\frac{2\ii(\mu^2-1)\mu}{(\mu^2+1)^2}\sigma_3+\widetilde V(\mu)$, where $\widetilde V(\mu)$ is as in \eqref{tilde-V} with $q=\beta_2$, $g_1=\gamma_1$, and $g_2=\gamma_2$.
\end{proof}

The next step is to demonstrate that the compatibility condition
\begin{equation}\label{compat}
\doublehat{U}_t - \doublehat{V}_y + [\doublehat{U},\doublehat{V}]=0
\end{equation}
yields the mCH equation in the $(y,t)$ variables, which is as follows:

\begin{proposition}
The mCH equation \eqref{mCH-2} in the $(y,t)$ variables reads as follows:
\begin{subequations}\label{ch_y}
\begin{align}\label{eq_y}
(\hat m^{-1})_t(y,t)&= 2\hat u_x(y,t),\\
\label{hatmy}
\hat m(y,t)&\coloneqq\hat u(y,t) - \hat u_{xx}(y,t) +1,
\end{align}
\end{subequations}
where $\hat f(y,t)\coloneqq\tilde f(x(y,t),t)$ for any function $\tilde f(x,t)$ and $x_y(y,t)=\hat m^{-1}(y,t)$.
\end{proposition}

\begin{proof}
Substituting $\tilde m_t=-(\tilde\omega\tilde m)_x$ from \eqref{mCH-2} and $x_t=\hat\omega$ from \eqref{x_t} into the equality
\[
\hat m_t(y,t)=\tilde m_x(x(y,t),t)x_t(y,t)+\tilde m_t(x(y,t),t)
\]
and using that $\tilde\omega_x=2\tilde m\tilde u_x$ we get
\[
\hat m_t(y,t)=\tilde m_x(x(y,t),t)\hat\omega(y,t)-\tilde m_x(x(y,t),t)\hat\omega(y,t)-2\tilde m^2(x(y,t),t)\hat u_x(y,t)=-2\hat u_x\hat m^2(y,t)
\]
and thus \eqref{eq_y} follows.
\end{proof}

\begin{remark}\label{rem-eqy2}
Notice that \eqref{hatmy} can be written as
\begin{equation}\label{eq_y-2}
\hat m(y,t)=\hat u(y,t) - \left(\hat u_x\right)_y(y,t) \hat m(y,t) +1.
\end{equation}
\end{remark}

Now, evaluating the compatibility equation \eqref{compat} at the singular points for $ \doublehat{U}$ and  $\doublehat{V}$, we get algebraic and differential equations amongst the coefficients of $ \doublehat{U}$ and  $\doublehat{V}$, i.e., amongst $\beta_1$, $\beta_2$, $\gamma_1$, and  $\gamma_2$, that can be reduced to \eqref{eq_y}.

\begin{proposition}
Let $\beta_1(y,t)$, $\beta_2(y,t)$, $\gamma_1(y,t)$, and $\gamma_2(y,t)$ be the functions determined in terms of $\hat M(y,t,\mu)$ as in Propositions \ref{prop-y} and \ref{prop-t}. Then they satisfy the following equations:
\begin{subequations}\label{rel}
\begin{align}\label{rel-a}
&\beta_{1 t}+\frac{\gamma_1+\gamma_2}{2}= 0;\\
\label{rel-b}
&\beta_2 - \frac{\gamma_2-\gamma_1}{2}=0;\\
\label{rel-c}
&(\gamma_1-\gamma_2)_y-(1+2\beta_1)(\gamma_1+\gamma_2)=0;\\
\label{rel-d}
&(\gamma_2+\gamma_1)_y+4\beta_1-(1+2\beta_1)(\gamma_1-\gamma_2)=0.
\end{align}
\end{subequations}
\end{proposition}

\begin{proof}
Recall $\beta_1$ and $\beta_2$ are given by \eqref{psi-1} and \eqref{psi-1-t}, respectively. Moreover, $\gamma_1\coloneqq 2a_2a_1$ and $\gamma_2\coloneqq -2a_3a_1^{-1}$, where $a_1$, $a_2$, and $a_3$ are defined by \eqref{hat-M-i}.

(i)
Evaluating the l.h.s. of \eqref{compat} as $\mu\to\infty$, the main term (of order $\ord(1)$) is
\[
\left(\beta_{1t}+\frac{\gamma_1+\gamma_2}{2}\right)\sigma_2,
\]
from which \eqref{rel-a} follows.

(ii)
Evaluating the l.h.s.\ of \eqref{compat} as $\mu\to 0$, the main term (of order $\ord(\mu^{-1})$)
is
\[
-\frac{1}{\mu}\left(\beta_2+\frac{\gamma_1-\gamma_2}{2}\right)\sigma_1,
\]
from which \eqref{rel-b} follows.

(iii)
Evaluating the l.h.s.\ of \eqref{compat} as $\mu\to 1$, the diagonal part of the main term (of order $\ord((\mu-1)^{-1})$) is
\[
\frac{\ii}{\mu-1}\left(\beta_{1t}-\beta_{2y}-\beta_1(\gamma_1+\gamma_2)\right)\sigma_3,
\]
from which \eqref{rel-c} follows, taking into account \eqref{rel-a} and \eqref{rel-b}.

(iv)
Evaluating the l.h.s.\ of \eqref{compat} as $\mu\to\ii$, the main term (of order $\ord((\mu-\ii)^{-1})$) is
\[
\frac{1}{\mu-\ii}\left\lbrack\begin{pmatrix}
0 & -\gamma_{1y}\\  -\gamma_{2y} & 0
\end{pmatrix}+ (1+2\beta_1)\begin{pmatrix}
0 & \gamma_{1}\\  -\gamma_{2} & 0
\end{pmatrix}  - 2\beta_1\begin{pmatrix}
0 &1 \\  1 & 0
\end{pmatrix}\right\rbrack,
\]
from which \eqref{rel-d} follows.
\end{proof}

\begin{proposition}\label{reduce}
Let $\hat m(y,t)$, $\hat u(y,t)$, and $x(y,t)$ be defined in terms of $\beta_1$, $\beta_2$, $\gamma_1$, and $\gamma_2$ as follows:
\begin{equation}\label{hmbbgg}
\hat m = (1+2\beta_1)^{-1},\quad\hat u=\beta_2=\frac{\gamma_2-\gamma_1}{2},\quad x_y=1+2\beta_1.
\end{equation}
Then the four equations \eqref{rel} reduce to \eqref{eq_y} and \eqref{eq_y-2}.
\end{proposition}

\begin{proof}
Indeed, defining $\hat u$ and $x(y,t)$ as prescribed in \eqref{hmbbgg}, equation \eqref{rel-c} implies that $\hat u_x=\hat u_yx_y^{-1}$ can be expressed as
\[
\hat u_x=-\frac{\gamma_1+\gamma_2}{2}.
\]
Then, taking into account the definition of $\hat m$ in \eqref{hmbbgg}, equation \eqref{rel-a} takes the form of the equation \eqref{eq_y}. Finally, using the notations introduced above, equation \eqref{eq_y-2} can be written as
\[
\frac{1}{1+2\beta_1} = \frac{\gamma_2-\gamma_1}{2}+  \frac{(\gamma_1+\gamma_2)_y}{2}\frac{1}{1+2\beta_1}+1,
\]
which is just equation \eqref{rel-d}.
\end{proof}

\begin{remark}\label{rem-u}
Formulas $\hat u = \frac{\gamma_2-\gamma_1}{2}$ and $\hat u_x = -\frac{\gamma_1+\gamma_2}{2}$ provide an alternative way to obtain $\hat u$ as well as $\hat u_x$ from the solution $\hat M$ of the RH problem. Indeed, according to Proposition \ref{prop-t}, $\hat u$ and $\hat u_x$ (as functions of $(y,t)$) can be obtained using the coefficients $a_j(y,t)$ (see \eqref{hat-M-i}) of the development of $\hat M(y,t,k)$ as $\mu\to\ii$ (thus avoiding the differentiations used in Section \ref{sec:recover}):
\begin{equation}\label{recsol}
\hat u(y,t)=-a_2a_1-a_3a_1^{-1},\qquad\hat u_x(y,t)=-a_2a_1+a_3 a_1^{-1},
\end{equation}
where $a_j(y,t)$ are determined by \eqref{hat-M-i}. Recall also the representation for $\hat m$ in terms of $\hat M$ evaluated as $\mu\to\infty$, see Proposition \ref{prop-y}:
\begin{equation}\label{m-sol}
\begin{split}
\hat m(y,t) &= \frac{1}{1+2\beta_1(y,t)}=\frac{1}{1-\eta(y,t)},\\ \eta(y,t)&\coloneqq\lim_{\mu\to\infty}\mu\hat M_{12}(y,t,\mu).
\end{split}
\end{equation}
Considered together with the expression for the change of variables \eqref{x(y,t)}, which can be written as (we indeed have $\hat\mu_1=a_1$ and $\hat\mu_2=a_1^{-1}$)
\begin{equation}\label{change}
x(y,t) = y + 2\ln a_1(y,t),
\end{equation}
equations \eqref{recsol} and \eqref{m-sol} give a parametric representation of the solution of the mCH equation \eqref{mCH-2}.
\end{remark}

\section{Solitons}\label{sec:8}
	
In the Riemann--Hilbert variant of the inverse scattering transform method, pure soliton solutions can be obtained  from the solutions of the RH problem assuming that the jump is trivial ($J\equiv I$), which reduces the construction to solving a system of linear algebraic equations generated by the residue conditions.

In order to construct the simplest, one-soliton solution, we consider the RH problem \eqref{Jp-y}--\eqref{sing-M-hat} with specific data, in particular $r(\mu)\equiv 0$, so that $\hat J\equiv I$. Regarding the other data, we require that $\hat M^{(1)}$ has a simple pole on the unit circle, at $\mu_1=\eul^{\ii\theta}$, $\theta\in(0,\frac{\pi}{2})$. It follows that $\hat M^{(1)}$ has also a simple pole at $\mu_2=-\eul^{-\ii\theta}=-\bar\mu_1=-\mu_1^{-1}$. According to the symmetries \eqref{sym-M-hat} the coefficients $\hat\varkappa_j(y,t)=\rho_j\eul^{-2\hat p(y,t,\mu_j)}$, $j=1,2$ in the residue conditions \eqref{res-M-hat} must satisfy the relations $\hat\varkappa_1=\overline{\hat\varkappa_2}=-\mu_1^{-2}\hat\varkappa_2$, that is, $\rho_1=\overline{\rho_2}=-\mu_1^{-2}\rho_2$ which imply $\rho_1=\ii\eul^{-\ii\theta}\hat\delta$ for some $\hat\delta\in\D{R}$. Further we denote $\hat\varkappa(y,t)\coloneqq\hat\varkappa_1(y,t)$ and $\rho\coloneqq\rho_1\in\D{C}$. So $\rho$ satisfies
\begin{equation}\label{ga}
\bar\rho=-\eul^{2\ii\theta}\rho.
\end{equation}
Thus we arrive at the following Riemann--Hilbert problem:

\begin{sol-rh-pb*}
Given $\theta\in(0,\frac{\pi}{2})$ and $\hat\delta\neq 0$ two real parameters, together with $c\in\D{R}$, find a piece-wise (w.r.t.~$\D{R}$) meromorphic, $2\times 2$-matrix valued function $\hat M(y,t,\mu)$ satisfying the following conditions:
\begin{enumerate}[\textbullet]
\item
The jump condition $\hat J\equiv I$ across $\D{R}$.
\item
The residue conditions \eqref{res-M-hat} at $\mu_1=\eul^{\ii\theta}$ and $\bar\mu_1=\eul^{-\ii\theta}$: 
\begin{subequations}\label{Res-sol}
\begin{align}\label{Res-sol-1}
\Res_{\eul^{\ii\theta}}\hat M^{(1)}(y,t,\mu)&=\frac{1}{\hat\varkappa(y,t)}\hat M^{(2)}(y,t,\eul^{\ii\theta}),\\
\label{Res-sol-2}
\Res_{\eul^{-\ii\theta}}\hat M^{(2)}(y,t,\mu)&=\frac{1}{\overline{\hat\varkappa(y,t)}}\hat M^{(1)}(y,t,\eul^{-\ii\theta}),
\end{align}
\end{subequations}
where $\hat\varkappa(y,t)=\ii\eul^{-\ii\theta}\hat\delta\eul^{-2\hat p(y,t,\eul^{\ii\theta})}$ with $\hat p(y,t,\eul^{\ii\theta})=\frac{\sin\theta}{2}(-y+\frac{2}{\cos^2\theta}t)$, and $\overline{\hat\varkappa}=-\eul^{2\ii\theta}\hat\varkappa$.
\item
The normalization condition $\hat M(y,t,\infty)=I$.
\item
The symmetries \eqref{sym-M-hat}.
\item
The singularity conditions \eqref{sing-M-hat} at $\mu=\pm1$.
\end{enumerate}
\end{sol-rh-pb*}

The residue conditions at $\mu_2$ and $\bar\mu_2$ follow from \eqref{Res-sol} using the symmetries \eqref{sym-M-hat}:
\begin{subequations}\label{Res-sol2}
\begin{align}\label{Res-sol2-1}
\Res_{-\eul^{-\ii\theta}}\hat M^{(1)}(y,t,\mu)&=\frac{1}{\overline{\hat\varkappa(y,t)}}\hat M^{(2)}(y,t,-\eul^{-\ii\theta}),\\
\label{Res-sol2-2}
\Res_{-\eul^{\ii\theta}}\hat M^{(2)}(y,t,\mu)&=\frac{1}{\hat\varkappa(y,t)}\hat M^{(1)}(y,t,-\eul^{\ii\theta}).
\end{align}
\end{subequations}

To summarize, the soliton RH problem of parameters $(\theta,\hat\delta)$ is the RH problem \eqref{Jp-y}--\eqref{sing-M-hat} with trivial jump condition and residue conditions data $\accol{\mu_j,\rho_j}_1^2$ where $\mu_1=-\bar\mu_2=\eul^{\ii\theta}$ and $\rho_1=\bar\rho_2=\ii\eul^{-\ii\theta}\hat\delta$.

\begin{remark}\label{rem:par}
Assume that the data of the soliton RH problem are associated with the spectral data corresponding to some initial data $u_0(x)$, see Section \ref{sec:spectral-data}. In particular, $b(\mu)\equiv 0$ and $a(\mu)$ has two zeros in $\D{C}^+$, each of multiplicity one, $\mu_1=\eul^{\ii\theta}$ and $\mu_2=-\eul^{-\ii\theta}$, both on the unit circle. The coefficient $\hat\varkappa$ in the residue condition for $M^{(1)}$ at $\mu_1$ is given by $\hat\varkappa=\rho\,\eul^{-2\hat p(y,t,\eul^{\ii\theta})}$ with $\rho=\dot a(\eul^{\ii\theta})\delta$, where the constant $\delta$ relates two Jost functions: $\hat\Phi_+^{(2)}(x,t,\mu_1)=\delta\hat\Phi_-^{(1)}(x,t,\mu_1)$. Using the symmetries \eqref{sym-Phi} and the relation $\bar\mu_1=\mu_1^{-1}$ we find that $\sigma_1\hat\Phi_\pm(\eul^{-\ii\theta})\sigma_1=\overline{\hat\Phi_\pm(\eul^{\ii\theta})}=\hat\Phi_\pm(\eul^{\ii\theta})$ and thus $\delta$ is real. Moreover, from the symmetry relation $a(\mu^{-1})=\overline{a(\bar\mu)}$ it follows that $\overline{\dot a(\eul^{\ii\theta})}=-\eul^{2\ii\theta}\dot a(\eul^{\ii\theta})$, and thus $\rho=\dot a(\eul^{\ii\theta})\delta$ satisfies \eqref{ga}. To conclude, in that case, $\hat\delta=-\ii\eul^{\ii\theta}\dot a(\eul^{\ii\theta})\delta$.
\end{remark}

\begin{proposition}\label{lem-M} 
Let $\theta\in(0,\frac{\pi}{2})$ and $\hat\delta\neq 0$ be two real parameters. Then, the soliton RH problem of parameters $(\theta,\hat\delta)$ has a solution $\hat M\equiv\hat M_{\theta,\hat\delta}$ provided that $c=1$:
\begin{align}\label{solM}
\hat M(y,t,\mu) 
&=I+\frac{\ii}{2}\frac{\hat\alpha_+(y,t)}{\mu-1}\begin{pmatrix}-1&1\\-1&1\end{pmatrix}-\frac{\ii}{2}\frac{\hat\alpha_+(y,t)}{\mu+1}\begin{pmatrix}1&1\\-1&-1\end{pmatrix}\notag\\
&\qquad
+\begin{pmatrix}
\frac{\ii\hat\kappa_1(y,t)\eul^{\ii\theta}}{\mu-\eul^{\ii\theta}}+\frac{\ii\hat\kappa_1(y,t)\eul^{-\ii\theta}}{\mu+\eul^{-\ii\theta}} & \frac{-\ii\hat\kappa_2(y,t)\eul^{-\ii\theta}}{\mu-\eul^{-\ii\theta}}+\frac{\ii\hat\kappa_2(y,t)\eul^{\ii\theta}}{\mu+\eul^{\ii\theta}}\\
\frac{\ii\hat\kappa_2(y,t)\eul^{\ii\theta}}{\mu-\eul^{\ii\theta}}+\frac{-\ii\hat\kappa_2(y,t)\eul^{-\ii\theta}}{\mu+\eul^{-\ii\theta}}&\frac{-\ii\hat\kappa_1(y,t)\eul^{-\ii\theta}}{\mu-\eul^{-\ii\theta}}+\frac{-\ii\hat\kappa_1(y,t)\eul^{\ii\theta}}{\mu+\eul^{\ii\theta}}
\end{pmatrix},
\end{align}
where
\begin{subequations}\label{kappa}
\begin{align}\label{kappa2}
\hat\kappa_2^{-1}(y,t)&=-\doublehat\varkappa(y,t)-\frac{\cos^2\theta}{4\doublehat\varkappa(y,t)\sin^2\theta}-\frac{1}{\sin\theta},\\\label{kappa1}
\hat\kappa_1(y,t)&=-\frac{\cos\theta}{2\doublehat\varkappa(y,t)\sin\theta}\hat\kappa_2(y,t),\\\label{alpha}
\hat\alpha_+(y,t)&=2\hat\kappa_2(y,t).
\end{align}
Here, 
\begin{equation}\label{varkappap}
\doublehat\varkappa(y,t)\coloneqq\hat\delta\,\eul^{-2\hat p(y,t,\eul^{\ii\theta})}\quad\text{with}\quad\hat p(y,t,\eul^{\ii\theta})=\frac{\sin\theta}{2}\left(-y+\frac{2}{\cos^2\theta}t\right).
\end{equation}
\end{subequations}
\end{proposition}

\begin{proof} 
Since $\hat M(\mu)\equiv\hat M(y,t,\mu)$ is solution of the soliton RH problem whose jump condition is trivial, it is a rational function, whose pole structure is specified by the singularity conditions \eqref{sing-M-hat} at $\mu=\pm1$ and by the residue conditions \eqref{Res-sol} at $\mu=\pm\eul^{\pm\ii\theta}$:
\begin{equation}\label{M-2}
\hat M(\mu) 
=I+\frac{\ii}{2}\frac{\hat\alpha_+}{\mu-1}\begin{pmatrix}-c & 1 \\-c & 1\end{pmatrix}-\frac{\ii}{2}\frac{\hat\alpha_+}{\mu+1}\begin{pmatrix}c & 1 \\-c & -1\end{pmatrix}
+\begin{pmatrix}\frac{c_1}{\mu-\eul^{\ii\theta}}+\frac{c_3}{\mu+\eul^{-\ii\theta}}&\frac{\tilde c_1}{\mu-\eul^{-\ii\theta}}+\frac{\tilde c_3}{\mu+\eul^{\ii\theta}}\\
\frac{c_2}{\mu-\eul^{\ii\theta}}+\frac{c_4}{\mu+\eul^{-\ii\theta}} &\frac{\tilde c_2}{\mu-\eul^{-\ii\theta}}+\frac{\tilde c_4}{\mu+\eul^{\ii\theta}}\end{pmatrix}
\end{equation}
with some $\hat\alpha_+(y,t)$, $c_j(y,t)$, $\tilde c_j(y,t)$, and $c$. We will specify the coefficients using the symmetries \eqref{sym-M-hat}. The symmetry $\hat M^{(1)}(-\mu)=\sigma_3\sigma_1\hat M^{(2)}(\mu)$ shows that $c=1$, $\tilde c_1=c_4$, $\tilde c_2=-c_3$, $\tilde c_3=c_2$, and $\tilde c_4=-c_1$. On the other hand, the symmetry $\hat M^{(1)}(-\bar\mu)=\sigma_3\overline{\hat M^{(1)}(\mu)}$ shows that $c_3=-\bar c_1$ and $c_4=\bar c_2$. Thus \eqref{M-2} takes the form
\[
\hat M(\mu) 
=I+\frac{\ii}{2}\frac{\hat\alpha_+}{\mu-1}\begin{pmatrix}-1 & 1 \\-1 & 1\end{pmatrix}-\frac{\ii}{2}\frac{\hat\alpha_+}{\mu+1}\begin{pmatrix}1 & 1 \\-1 & -1\end{pmatrix}
+\begin{pmatrix}\frac{c_1}{\mu-\eul^{\ii\theta}}+\frac{-\bar c_1}{\mu+\eul^{-\ii\theta}}&\frac{\bar c_2}{\mu-\eul^{-\ii\theta}}+\frac{c_2}{\mu+\eul^{\ii\theta}}\\
\frac{c_2}{\mu-\eul^{\ii\theta}}+\frac{\bar c_2}{\mu+\eul^{-\ii\theta}} &\frac{\bar c_1}{\mu-\eul^{-\ii\theta}}+\frac{-c_1}{\mu+\eul^{\ii\theta}}\end{pmatrix}.
\]
The symmetry $\hat M^{(1)}(-\mu^{-1})=\sigma_3\hat M^{(1)}(\mu)$ shows that $c_3=c_1\eul^{-2\ii\theta}$ and $c_4=-c_2\eul^{-2\ii\theta}$, so that $\bar c_j=-c_j\eul^{-2\ii\theta}$ for $j=1,2$, that is, $c_j(y,t)=\ii\eul^{\ii\theta}\hat\kappa_j(y,t)$ with $\hat\kappa_j(y,t)\in\D{R}$. Thus we get \eqref{solM}.

Then, using $\hat M(0)=\sigma_1\hat M(\infty)\sigma_1=I$, it follows that $\hat\alpha_+=2\hat\kappa_2$, that is, \eqref{alpha}. Introducing $\doublehat\varkappa(y,t)\coloneqq\hat\delta\,\eul^{-2\hat p(y,t,\eul^{\ii\theta})}$ so that $\hat\varkappa(y,t)=\ii\eul^{-\ii\theta}\doublehat\varkappa(y,t)$ and substituting \eqref{solM} into the residue condition \eqref{Res-sol-1} at $\eul^{\ii\theta}$, we find \eqref{kappa1} on the first row and then \eqref{kappa2} on the second one.
\end{proof}
 
\begin{remark}
Assume that the data of our soliton RH problem are derived from the spectral data corresponding to some initial data $u_0(x)$, as in Remark \ref{rem:par}. Then, it directly follows that $c=1$. Since $b(\mu)\equiv 0$ we indeed have (see Remark \ref{rem:sing} and \eqref{c}) $\rho=0$, $b_1=0$, and $a_1^2=1$; thus $c=1$.
\end{remark}

According to Section~\ref{sec:7}, a solution of the soliton RH problem gives rise to a solution (at least, locally, in the $(y,t)$ variables) of the mCH equation. Thus, Proposition~\ref{lem-M} provides a family of one-soliton solutions parametrized by two real parameters $\theta\in(0,\frac{\pi}{2})$ and $\hat\delta\neq 0$. 

\begin{proposition}
The one-soliton solution $\hat u\equiv\hat u_{\theta,\hat\delta}$ of parameters $(\theta,\hat\delta)$ has the following form in the $(y,t)$-scale:
\begin{subequations}\label{usolitonz}
\begin{equation}\label{usoliton}
\hat u(y,t)=4\tan^2\theta\,\frac{z^2(y,t)+2\cos^2\theta\cdot z(y,t)+\cos^2\theta}{(z^2(y,t)+2z(y,t)+\cos^2\theta)^2}z(y,t),
\end{equation}
where
\begin{equation}\label{z}
z(y,t)=2\hat\delta\sin\theta\,\eul^{\sin\theta\left(y-\frac{2}{\cos^2\theta}t\right)}.
\end{equation}
\end{subequations}
\end{proposition}

\begin{proof}
Let $z(y,t)$ be defined by
\begin{equation}\label{zdef}
z(y,t)\coloneqq 2\doublehat\varkappa(y,t)\sin\theta.
\end{equation}
Then, $z(y,t)=2\hat\delta\sin\theta\,\eul^{\sin\theta\left(y-\frac{2}{\cos^2\theta}t\right)}$. Thus, $z$ is real-valued. Moreover, $z(y,t)>0$ if $\hat\delta>0$ and $z(y,t)<0$ if $\hat\delta<0$. Using \eqref{kappa2}, \eqref{kappa1}, and \eqref{zdef} we get the following expressions of $\hat\kappa_2$ and $\hat\kappa_1$:
\begin{equation}\label{kappa12z}
\hat\kappa_2=-\frac{2z\sin\theta}{z^2+2z+\cos^2\theta}\quad\text{and}\quad\hat\kappa_1=-\frac{\cos\theta}{z}\hat\kappa_2=\frac{2\sin\theta\cos\theta}{z^2+2z+\cos^2\theta}.
\end{equation}
In order to obtain the formula for the soliton solution $\hat u\equiv\hat u(y,t)$, we use the relation
\begin{equation}\label{recsolu}
\hat u=-a_2a_1-a_3a_1^{-1}
\end{equation}
from \eqref{recsol}. To compute $a_1\equiv a_1(y,t)$ we observe that $a_1=\hat M_{11}(\ii)$. We thus obtain
\[
a_1=1-\frac{\hat\alpha_+}{2}-\ii\kappa_1\frac{1+\eul^{2\ii\theta}}{2(1-\sin\theta)}=1-\hat\kappa_2+\hat\kappa_1\frac{\cos\theta}{1-\sin\theta},
\] 
using the relation $\frac{\hat\alpha_+}{2}=\hat\kappa_2$ from \eqref{alpha}. Using the expressions of $\hat\kappa_1$ and $\hat\kappa_2$ from \eqref{kappa12z} we get
\begin{subequations}\label{a123}
\begin{equation}\label{a_1}
a_1=\frac{z+1+\sin\theta}{z+1-\sin\theta}.
\end{equation}
To compute $a_2\equiv a_2(y,t)$ and $a_3\equiv a_3(y,t)$ we observe that $a_2=\partial_\mu\hat M_{12}(\ii)$ and $a_3=\partial_\mu\hat M_{21}(\ii)$. Using in addition the expression of $\hat\kappa_2$ from \eqref{kappa12z} we obtain
\begin{align}\label{a_2}
a_2&=\frac{\sin\theta}{1+\sin\theta}\hat\kappa_2=-\frac{2z\sin^2\theta}{(1+\sin\theta)(z^2+2z+\cos^2\theta)},\\
\label{a_3}
a_3&=\frac{\sin\theta}{1-\sin\theta}\hat\kappa_2=-\frac{2z\sin^2\theta}{(1-\sin\theta)(z^2+2z+\cos^2\theta)}.
\end{align}
\end{subequations}
Then, substituting \eqref{a123} into \eqref{recsolu}, we arrive at \eqref{usoliton}.
\end{proof}

It follows from \eqref{usoliton} that if $\hat\delta>0$, then for any $t\geq 0$, $\hat u(y,t)$ is a smooth function of $y$ having a single peak and (exponentially) approaching $0$ as $y\to\pm\infty$. On the other hand, if $\hat\delta<0$, then $\tilde u$ has two singular points
corresponding to $z=-1\pm\sin\theta$.

Now let us discuss the change of variable $(y,t)\mapsto(x,t)$, which can be specified explicitly. This change of variable is associated with $\tilde u_{\theta,\hat\delta}$, that is, it is given by \eqref{x(y,t)} where $\hat\mu_1$ and $\hat\mu_2$ are defined in terms of $\hat M\equiv\hat M_{\theta,\hat\delta}$.

\begin{proposition}\label{lem-change-solitons} 
The change of variable $x(y,t)$ associated with the soliton $\tilde u_{\theta,\hat\delta}$ takes the following form:
\begin{equation}\label{x(y)sol}
x(y,t)=y+2\ln\frac{z(y,t)+1+\sin\theta}{z(y,t)+1-\sin\theta}.
\end{equation}
\end{proposition}

\begin{proof}
As we have shown in Section \ref{sec:7}, $x(y,t)$ can be given by \eqref{change}:
\begin{equation}\label{change-2}
x(y,t)=y+2\ln a_1(y,t),
\end{equation}
where $a_1(y,t)=\hat M_{11}(y,t,\ii)$. Substituting \eqref{a_1} into \eqref{change-2}, we obtain \eqref{x(y)sol}.
\end{proof}

\begin{corollary}
Let $x(y,t)$ be the change of variable associated with $\tilde u_{\theta,\hat\delta}$. Its regularity properties are as follows.
\begin{enumerate}[\rm(a)]
\item
If $\hat\delta<0$, then $x(\,\cdot\,,t)$ is singular: there exist values of $y$ at which $x(y,t)$ is infinite.
\item
If $\hat\delta>0$, then $x(\,\cdot\,,t)\colon\D{R}\to\D{R}$ is a regular map. Moreover, it has the following additional properties:
\begin{enumerate}[\rm(i)]
\item
If $\theta\in(0,\frac{\pi}{3})$, then $x(\,\cdot\,,t)\colon\D{R}\to\D{R}$ is a diffeomorphism for any $t\geq 0$.
\item
If $\theta=\frac{\pi}{3}$, then $x(\,\cdot\,,t)\colon\D{R}\to\D{R}$ is a bijection, but the derivative of the inverse map has a singularity, and only one.
\item
If $\theta\in(\frac{\pi}{3},\frac{\pi}{2})$, then $x(\,\cdot\,,t)$ is not monotonous. More precisely, there are three intervals of monotonicity.
\end{enumerate}
\end{enumerate}
\end{corollary}

The possible singularities of $x(y,t)$ are those for $\hat u(y,t)$: they correspond to $z=-1\pm\sin\theta$. Therefore, if $\hat\delta>0$, then $z(y,t)>0$ and thus there are no singularities, whereas if $\hat\delta< 0$, then $x(y,t)$ is singular at those $y$ where $z=-1\pm\sin\theta$.

We now consider the case $\hat\delta>0$ (and thus $z(y,t)>0$). The derivative $\partial_yx(y,t)\equiv x_y(y,t)$ is given by
\begin{equation}\label{xyR}
x_y(y,t)=R(z(y,t)),\text{ where }R(z)=\frac{z^2+2z\cos2\theta+\cos^2\theta}{z^2+2z+\cos^2\theta}.
\end{equation}
It follows that $R(0)=R(\infty)=1$. Moreover, we have the following:
\begin{enumerate}[1)]
\item
If $\theta\in(0,\frac{\pi}{3})$, then $R(z)>0$ for all $z\geq 0$.
\item
If $\theta=\frac{\pi}{3}$, then $z=\frac{1}{2}$ is a double zero of $R(z)$.
\item
If $\theta\in(\frac{\pi}{3},\frac{\pi}{2})$, then
\begin{enumerate}[a)]
\item
$R(z)>0$ for $z\in [0,-\cos2\theta-\sqrt{-\sin\theta\cdot\sin3\theta})\cup(-\cos2\theta+\sqrt{-\sin\theta\cdot\sin3\theta}),+\infty)$,
\item
$R(z)<0$ for $z\in(-\cos2\theta-\sqrt{-\sin\theta\cdot\sin3\theta},-\cos2\theta+\sqrt{-\sin\theta\cdot\sin3\theta})$.
\end{enumerate}
\end{enumerate}
It follows that for $\theta\in(0,\frac{\pi}{3})$ the solution is smooth (both in the $(y,t)$ and the $(x,t)$ variables). On the other hand, for $\theta=\frac{\pi}{3}$ the solution $\tilde u(x,t)=\hat u(y(x,t),t)$ is given in parametric form by
\begin{subequations}\label{fsmooth}
\begin{align}\label{fsmooth-u}
\hat u(y,t)&=48z(y,t)\frac{4z^2(y,t)+2z(y,t)+1}{(4z^2(y,t)+8z(y,t)+1)^2},\\
\label{fsmooth-z}
z(y,t)&=\hat\delta\sqrt{3}\,\eul^{\frac{\sqrt{3}}{2}y}\eul^{-4\sqrt{3}t},\\
\label{fsmooth-x}
x(y,t)&=y+2\ln\frac{\hat\delta\sqrt{3}\,\eul^{\frac{\sqrt{3}}{2}y}\eul^{-4\sqrt{3}t}+1+\frac{\sqrt{3}}{2}}{\hat\delta\sqrt{3}\,\eul^{\frac{\sqrt{3}}{2}y}\eul^{-4\sqrt{3}t}+1-\frac{\sqrt{3}}{2}}.
\end{align}
\end{subequations}
In particular, in the latter case \eqref{xyR} and \eqref{fsmooth-u} give 
\[
x_y=\frac{2\hat z^2}{2\hat z^2+6\hat z+3}\quad\text{and}\quad\hat u_y=-24\sqrt{3}\,\frac{\hat z^3(\hat z+1)(2\hat z+1)}{(2\hat z^2+6\hat z+3)^3}\,,
\]
where $\hat z\coloneqq z-\frac{1}{2}$. Thus $x_y$ has a double zero at $\hat z=0$, which corresponds to the crest of the solution, whereas, at the same point, $\hat u_y$ has a triple zero, so that $\tilde u_x=\hat u_y/x_y=0$. Consequently, $\tilde u(x,t)$ is still continuous, with a continuous first derivative $\tilde u_x$ that vanish at the crest, but the higher order derivatives become unbounded at this point, e.g., $\tilde u_{xx}\sim-\frac{3}{2}\,\hat z^{-2}$ as $\hat z\to 0$. This unusual (finite) smoothness property of the soliton corresponding to the parameters separating (infinitely) smooth solitons from multivalued solutions (associated with the breaking of bijectivity of $x(\,\cdot\,,t)\colon\D{R}\to\D{R}$) was first reported by Matsuno \cite{M13}, where the soliton solutions were constructed using a direct method. 

Thus we arrive at the following description of the one-soliton solutions (consistent with \cite{M13}*{see (3.4) and (3.14)}):

\begin{theorem} \label{thm:one-soliton}
The mCH equation in the form \eqref{mCH2} has a family of one-soliton solutions, regular as well as non-regular, $\tilde u(x,t)\equiv\tilde u_{\theta,\hat\delta}(x,t)$, parametrized by two parameters, $\hat\delta>0$ and $\theta\in(0,\frac{\pi}{2})$. These solitons $\tilde u(x,t)\equiv\hat u(y(x,t),t)$ are given, in parametric form, by
\begin{subequations}\label{solitons}
\begin{align}\label{solution}
\hat u(y,t)&=4\tan^2\theta\frac{z^2(y,t)+2\cos^2\theta\cdot z(y,t)+\cos^2\theta}{(z^2(y,t)+2z(y,t)+\cos^2\theta)^2}z(y,t),\\
\label{changesol}
x(y,t)&=y+2\ln\frac{z(y,t)+1+\sin\theta}{z(y,t)+1-\sin\theta},\\
z(y,t)&=2\hat\delta\sin\theta\,\eul^{y\sin\theta}\eul^{-\frac{2\sin\theta}{\cos^2\theta}t}.
\end{align}
\end{subequations}
They have different properties depending on the value of the parameter $\theta$:
\begin{enumerate}[\rm(i)]
\item
For $\theta\in(0,\frac{\pi}{3})$, the one-soliton solution $\tilde u(x,t)$ is smooth in the $(x,t)$ variables.
\item
For $\theta=\frac{\pi}{3}$, then $\tilde u(x,t)$ is given by \eqref{fsmooth} and has finite smoothness: $u$ and $u_x$ are continuous with $\tilde u_x(x,t)=0$ at the crest when $z(y(x,t),t)=\frac{1}{2}$, but near the crest the higher derivatives become unbounded as $z\to\frac{1}{2}$.
\item
If $\theta\in(\frac{\pi}{3},\frac{\pi}{2})$, then $\tilde u(x,t)=\hat u(y,t)$ is regular in $(y,t)$, multivalued in $(x,t)$, and loop-shaped.
\end{enumerate}
\end{theorem}

\section*{Acknowledgment}
We would like to thank Y.~Matsuno for pointing out an error in item (ii) of the last theorem in the first version of this paper.

\end{document}